\documentclass[draftclsnofoot,onecolumn]{IEEEtran}
\usepackage{pifont}
\usepackage{graphicx}
\usepackage{fancyhdr}
\usepackage{amsmath,amssymb,amstext,amsfonts}
\usepackage{bm}
\usepackage{algorithm}
\usepackage{array}
\usepackage{booktabs}
\usepackage{cite}

\usepackage{subfigure}
\usepackage{balance}
\usepackage{algorithmicx}
\usepackage{algpseudocode}

\usepackage{braket}
\usepackage{tabulary}
\usepackage{diagbox}
\usepackage{amsthm}
\makeatletter
\renewcommand{\maketag@@@}[1]{\hbox{\m@th\normalsize\normalfont#1}}%
\makeatother

\newtheorem{proposition}{\textbf{Proposition}}
\newtheorem{theorem}{Theorem}
\newtheorem{lemma}{Lemma}
\newtheorem{remark}{Remark}
\newtheorem{corollary}{Corollary}

\begin{document}

\title{Finite-Support Structure {of Optimal i.i.d. Inputs for} Finite-Memory Poisson Channels}

\author{Renzhi Yuan, \IEEEmembership{Member, IEEE} and Mugen Peng, \IEEEmembership{Fellow, IEEE}
\thanks{The authors are with the State Key Laboratory of Networking and Switching Technology, Beijing University of Posts and Telecommunications, Beijing 100876, China, and also with Beijing Key Laboratory of Convergent Communications and Networking Technologies in LEO Satellite Systems, Beijing 100876, China. This work is supported by Beijing Natural Science Foundation under Grant No. 4262010.}%
\thanks{Corresponding Author: Mugen Peng (pmg@bupt.edu.cn)}
\thanks{A journal version of this work is undergoing peer review in IEEE Transactions on Information Theory.}
}

\maketitle

\begin{abstract}
Discrete-time Poisson channels with finite intersymbol interference provide a natural model for direct-detection optical links in which multipath memory and signal-dependent shot noise appear simultaneously. Under peak and average optical-intensity constraints, we study the stationary mutual-information-rate optimization under an independent and identically distributed (i.i.d.) input restriction. We prove that every input distribution maximizing the stationary mutual information rate within the i.i.d. input class has finite support. The proof is carried out directly on the entropy rate of the continuous-state hidden Markov output process induced by the finite-memory channel. We first establish a filtering-forgetting estimate whose constants are uniform over all admissible i.i.d. input laws. We then derive the entropy-rate first variation, construct a holomorphic extension of the corresponding influence function, and combine the Karush-Kuhn-Tucker condition with a supralinear growth argument.
\end{abstract}


\section{Introduction}

\subsection{Background and Motivation}
In recent years, optical wireless communication has emerged as a core technology for next-generation networks, owing to its unlicensed spectrum and immunity to electromagnetic interference. Typical examples include ultraviolet non-line-of-sight scattering communications~\cite{wang2026integrated,yuan2019monte,wang2023inter,yuan2026turbulent} and highly sensitive single-photon detection-based underwater wireless optical communications~\cite{zeng2017survey}. In these links, the received photons may experience strong multipath scattering before arriving at the receiver. The resulting intersymbol interference (ISI) has to be considered together with the signal-dependent shot noise caused by direct photodetection.

A useful discrete-time abstraction for this situation is the finite-memory Poisson channel~\cite{wyner1988capacity,lapidoth2009discrete_poisson}. The input intensity is usually constrained by both the peak optical power and the average optical power because of eye-safety regulations, hardware limitations, and energy-consumption requirements. If arbitrary temporally correlated inputs are permitted, the capacity problem becomes a continuous-state hidden Markov model (HMM) optimization with an infinite-dimensional input process. In this work, we focus on a more implementable and analytically tractable signaling class, where the channel inputs are independent and identically distributed (i.i.d.). Our objective is to establish a structural property of the marginal input laws that maximize the stationary mutual information rate under the i.i.d. restriction for finite-memory Poisson channels.

\subsection{Related Works}
The information-theoretic study of Poisson channels was initiated mainly from continuous-time direct-detection models. Davis~\cite{davis1980poisson_type} analyzed the capacity and cutoff rate of Poisson-type channels under peak and average constraints, and Wyner~\cite{wyner1988capacity,wyner1988capacity_part2} derived the capacity and error exponent of the direct-detection photon channel. Frey~\cite{frey1991poisson_capacity,frey1992lp_norm} further considered random or time-varying noise intensities and $L_p$-norm input constraints. Practical waveform restrictions were incorporated in later works. For example, Bar-David and Kaplan~\cite{bardavid1984oppm} studied photon-limited overlapping pulse-position modulation under pulsewidth constraints. Shamai~\cite{shamai1990capacity} investigated pulse-amplitude-modulated direct-detection photon channels and showed that the optimizing symbol distribution can be finite-valued in several regimes. Binary inputs with intertransition constraints were studied in~\cite{shamai1991intertransition}, and spectrally constrained Poisson channels were considered in~\cite{shamai1993spectrally_constrained}. Multiuser and multi-aperture Poisson models were studied in~\cite{lapidoth1998poisson_mac} and~\cite{haas2002mimo_poisson}, respectively.

For discrete-time memoryless Poisson channels, many studies have focused on capacity bounds and the structure of capacity-achieving input distributions. The discreteness of optimal distributions for pulse-amplitude-modulated Poisson models was established in~\cite{shamai1990capacity}. Lapidoth and Moser~\cite{lapidoth2009discrete_poisson} derived asymptotic capacity results under peak and average constraints. Cheraghchi and Ribeiro~\cite{cheraghchi2018improved_isit,cheraghchi2019improved} improved capacity upper bounds and studied structural properties under average and peak constraints. Cao, Hranilovic, and Chen~\cite{cao2014dtp_part1,cao2014dtp_part2} developed general properties, numerical methods, and binary-input optimality conditions for discrete-time Poisson channels. More recently, Dytso, Barletta, and Shamai~\cite{dytso2021poisson_support,barletta2024poisson_support} further sharpened support-size results for amplitude-constrained Poisson noise channels. Discreteness phenomena for other amplitude-constrained channels were studied in the classical work of Smith~\cite{smith1971information} and in subsequent studies~\cite{tchamkerten2004discreteness,behboodi2017censored}. Optical intensity channels with Gaussian-noise approximations were also studied under power and bandwidth constraints~\cite{hranilovic2004capacity}, which provide useful comparisons to Poisson-noise models.

The introduction of memory leads to a different type of problem. For finite-state channels, Blahut-Arimoto-type recursions and their extensions provide computable tools for multiletter optimizations~\cite{blahut1972computation,arimoto1972algorithm,vontobel2008generalization}. In general channels with memory, capacity is often described by multiletter or directed-information expressions, and temporally correlated or feedback-dependent inputs may be essential~\cite{permuter2009deterministic_feedback,permuter2008trapdoor_feedback}. In the finite-memory Poisson channel considered here, the hidden state becomes continuous whenever the input law is not finite-alphabet, and the stationary mutual information rate depends on the entire observation past. A point-mass perturbation of the input law therefore affects a finite physical output window first and then propagates through the nonlinear predictive filter of the HMM output process~\cite{han2006analyticity,han2013continuous_hmm_analyticity,tadic2019continuous_hmm_analyticity}. Related finite-memory Poisson models have also appeared in diffusion-based molecular communications. Aminian et al.~\cite{aminian2015lti_poisson} introduced the linear time-invariant Poisson channel, gave finite-letter capacity characterizations, and derived upper and lower bounds. Etemadi et al.~\cite{etemadi2019compound_poisson} considered discrete-time compound Poisson models arising from stochastic biological noise sources, and Ratti et al.~\cite{ratti2021bounds} derived constrained-capacity bounds for diffusive Poisson molecular channels with memory. Strong data-processing inequalities~\cite{raginsky2016sdpi} and neural estimators of directed information~\cite{aharoni2020dine,tsur2023neural_di,aharoni2020rl_fsc,aharoni2024neural_polar} provide complementary ways to analyze or approximate information measures in channels with memory. However, these works do not give a variational structural theorem for the optimizing continuous input distribution of finite-memory Poisson channels.

Different from the unrestricted Shannon capacity of channels with memory, this paper studies the stationary mutual information rate under the i.i.d. input restriction; no separate operational coding-capacity claim is made. This restriction is deliberate. It matches common finite-alphabet modulation practice and isolates the structural question of the optimal marginal law. The main problem is whether an optimizing i.i.d. marginal distribution can contain a continuous component, or whether finite-alphabet signaling is sufficient even in the presence of finite ISI.

\subsection{Our Contributions}
In this work, we prove that every input law maximizing the stationary mutual information rate under the i.i.d. restriction for the finite-memory Poisson channel has finite support. The main contributions are summarized as follows.
\begin{itemize}
 \item We formulate the variational optimality condition directly for the stationary mutual information rate of the finite-memory Poisson channel. Since the output process is a continuous-state HMM, the output term is an entropy rate rather than a single-letter entropy. The resulting point-mass influence function includes both the direct finite-window effect of one input symbol and its later influence through the predictive filter.
 \item We establish a filtering-forgetting estimate for the continuous-state HMM induced by every admissible i.i.d. input law. The estimate is uniform in the input distribution. Its proof uses the block structure of the finite-memory channel, the common dominating terminal measure $F^{\otimes N}$ for one full memory block, together with a Hilbert-projective-metric/Birkhoff-contraction argument for the unnormalized block kernels. This result gives the continuity and locally uniform convergence needed for the entropy-rate first variation.
 \item We prove that the entropy-rate influence function has a holomorphic extension to a right-half complex domain. We then show that this extension satisfies a supralinear growth bound on the positive real axis. Together with the Karush-Kuhn-Tucker (KKT) condition and the identity theorem for holomorphic functions, this rules out infinitely many support points of any rate-maximizing i.i.d. distribution.
\end{itemize}

The rest of this paper is organized as follows. Section~\ref{sec:system_model} introduces the finite-memory Poisson channel model and the i.i.d.-constrained mutual-information-rate formulation. Section~\ref{sec:main_results} presents the finite-support theorem. Section~\ref{sec:proof_finite_support} proves the theorem by using filtering-forgetting, entropy-rate first variation, KKT optimality, and a supralinear growth contradiction. Section~\ref{sec:numerical_results} presents finite-grid numerical studies, and Section~\ref{sec:conclusion} concludes the paper.

\section{System Model and Mutual Information Rate Formulation} \label{sec:system_model}

\subsection{Discrete-Time Poisson Channel with Finite Memory}\label{subsec:HMM}
We consider a discrete-time baseband equivalent model for a direct-detection optical channel with finite-order ISI~\cite{wang2023inter}. The input process is denoted by $\{X_k\}_{k\in\mathbb Z}$, where $X_k$ represents the expected emitted photon intensity in the $k$-th time slot. Throughout the paper, $\mathcal P(\mathcal A)$ denotes the set of Borel probability measures on a compact metric space $\mathcal A$, $\mathcal B(\mathcal A)$ denotes its Borel $\sigma$-field, $\operatorname{supp}(F)$ denotes the closed support of a probability law $F$, and all logarithms are natural. For integers $i\le j$, $X_i^j\triangleq (X_i,\ldots,X_j)$ and $Y_i^j\triangleq (Y_i,\ldots,Y_j)$. The input is subject to a peak optical-intensity constraint and an average optical-intensity constraint. For a marginal input law $F$ on $[0,A]$, define the admissible set
\begin{equation}
\label{eq:admissible_set_secII}
\begin{aligned}
 \mathcal P_A(P_{\rm avg})
 \triangleq
 \left\{
 F\in\mathcal P([0,A]):
 \int_{0}^{A}x\,dF(x)\le P_{\rm avg}
 \right\}.
\end{aligned}
\end{equation}
Equivalently, under the i.i.d. input law induced by $F$, the peak and average optical-intensity constraints can be written as
\begin{equation}
\label{eq:power_const}
\begin{cases}
 0\le X_k\le A, \\
 \mathbb E_F[X_k]\le P_{\rm avg} .
\end{cases}
\end{equation}

The multipath dispersion of the optical channel is characterized by a finite memory length $L$. The channel impulse response is defined by nonnegative tap coefficients $h_0,h_1,\ldots,h_L$, where $h_0>0$ denotes the main line-of-sight path; we write $\mathbf h\triangleq(h_0,\ldots,h_L)$. To decouple the theoretical analysis from the absolute spatial path loss, we assume
\begin{equation}
\label{eq:tap_normalization_secII}
\begin{aligned}
 \sum_{\ell=0}^{L}h_\ell=1 .
\end{aligned}
\end{equation}
At time $k$, the received optical intensity, including the constant background dark current $\lambda_0>0$, is
\begin{equation}
\label{eq:intensity_secII}
\begin{aligned}
 Z_k
 \triangleq
 \lambda_0+\sum_{\ell=0}^{L}h_\ell X_{k-\ell} .
\end{aligned}
\end{equation}
Equivalently, $Z_k=h_0X_k+W_k$, where $W_k=\lambda_0+\sum_{\ell=1}^{L}h_\ell X_{k-\ell}$ collects the causal ISI from the past $L$ symbols and the physical background noise. We shall use
\begin{equation}
\label{eq:BA_def}
\begin{aligned}
 B_A\triangleq \lambda_0+A
\end{aligned}
\end{equation}
for the uniform upper bound on all feasible Poisson means.

Conditioned on the input sequence, the photon-count observations $\{Y_k\}_{k\in\mathbb Z}$ are conditionally independent Poisson random variables. In particular,
\begin{equation}
\label{eq:poisson_channel_secII}
\begin{aligned}
 \Pr_F(Y_k=y\mid X_{k-L}^{k})
 &=\frac{e^{-Z_k}Z_k^y}{y!},
 \qquad y\in\mathbb N_0,
\end{aligned}
\end{equation}
where $\mathbb N_0\triangleq\{0,1,2,\ldots\}$. Because the state $(X_{k-L},\ldots,X_k)$ takes values in a continuum whenever $F$ is not finite-alphabet, the output process is naturally viewed as a continuous-state hidden Markov process.

\subsection{Mutual Information Rate Formulation}
For $F\in\mathcal P([0,A])$, let $\mathbb P_F$, $\mathbb E_F$, $H_F$, and $I_F$ denote probability, expectation, entropy, and mutual information evaluated under the stationary i.i.d. input process with marginal law $F$ and the channel law in \eqref{eq:poisson_channel_secII}. For a block length $n$, the stationary mutual information rate functional is
\begin{equation}
\label{eq:mi_rate_secII}
\begin{aligned}
 \mathcal I(F)
 &\triangleq
 \lim_{n\to\infty}\frac{1}{n} I_F(X_1^n;Y_1^n) = h_F(\mathbf Y)-h_F(\mathbf Y|\mathbf X),
\end{aligned}
\end{equation}
where the output entropy rate and the conditional entropy rate are
\begin{equation}
\label{eq:entropy_rates_secII}
\begin{cases}
 h_F(\mathbf Y)
 \triangleq H_F(Y_0|Y_{-\infty}^{-1}), \\
 h_F(\mathbf Y|\mathbf X)
 \triangleq \mathbb E_F[\phi(Z_0)] .
\end{cases}
\end{equation}
Here $Y_{-\infty}^{-1}\triangleq (\ldots,Y_{-2},Y_{-1})$. The second equality in \eqref{eq:mi_rate_secII} follows from the chain rule of entropy, stationarity, and the conditional independence of the Poisson observations given the input sequence~\cite{cover2006elements}. More precisely, for a finite block the first $L$ outputs also depend on the pre-block inputs $X_{1-L}^{0}$; their total contribution is $O(1)$ because all Poisson means lie in $[\lambda_0,B_A]$. Hence $n^{-1}H_F(Y_1^n|X_1^n)=\mathbb E_F[\phi(Z_0)]+O(n^{-1})$, which yields \eqref{eq:mi_rate_secII} after $n\to\infty$. The single-letter function $\phi(z)$ denotes the Shannon entropy of a Poisson random variable with mean $z$, i.e.,
\begin{equation}
\label{eq:poisson_entropy_fn}
\begin{aligned}
 \phi(z)
 \triangleq
 z-z\log z+
 \sum_{m=0}^{\infty}\frac{e^{-z}z^m}{m!}\log(m!) .
\end{aligned}
\end{equation}

The i.i.d.-constrained stationary mutual-information-rate optimum under the joint peak and average constraints is defined as
\begin{equation}
\label{eq:capacity_def}
\begin{aligned}
 R_{\rm i.i.d.}^{\rm ISI}(A,P_{\rm avg})
 \triangleq
 \sup_{F\in\mathcal P_A(P_{\rm avg})}\mathcal I(F) .
\end{aligned}
\end{equation}
The quantity in \eqref{eq:capacity_def} is an information-rate optimization and is not asserted here to equal an operational coding capacity. The main analytical difficulty in \eqref{eq:capacity_def} lies in the term $h_F(\mathbf Y)$, which contains the infinite observation history of a continuous-state hidden Markov output process. The next sections address this difficulty through a variational analysis of $\mathcal I(F)$.

\section{Main Theoretical Result} \label{sec:main_results}

The main result is stated as follows. Its proof is given in Section~\ref{sec:proof_finite_support}.

\begin{theorem}
\label{thm:finite_support}
Consider the finite-memory Poisson channel $Y_k|X_{k-L}^k\sim\operatorname{Poisson}\left(\lambda_0+\sum_{\ell=0}^{L}h_\ell X_{k-\ell}\right)$, where $\lambda_0>0$, $h_\ell\geq 0$, $h_0>0$, and $\sum_{\ell=0}^Lh_\ell=1$. Under the peak and average constraints $0\leq X_k\leq A$ and $\mathbb E[X_k]\leq P_{\rm avg}$, every input distribution $F^\star$ that maximizes the stationary mutual information rate over $\mathcal P_A(P_{\rm avg})$ has finite support on $[0,A]$.
\end{theorem}

{
\begin{remark}
Theorem~\ref{thm:finite_support} is qualitative: the identity-theorem argument proves finiteness but does not provide a parameter-dependent upper bound on $|\operatorname{supp}(F^\star)|$ or an explicit procedure for locating the mass points. Its computational implication is therefore structural rather than finite-dimensional in a prescribed cardinality: it rules out continuous components and justifies finite-alphabet approximations, while a general quantitative support-size bound is not available from the present identity-theorem argument.
\end{remark}
}

Theorem~\ref{thm:finite_support} gives a structural reduction of the i.i.d.-constrained optimization problem. Although the admissible set contains all probability laws on the continuous interval $[0,A]$, no maximizing law can have a continuous component. Hence any optimal i.i.d. signaling law is finite-alphabet, although the locations and probabilities of its mass points still depend on the channel parameters and the power constraints. The statement holds for every finite peak level $A$ and every positive dark current $\lambda_0$.

{
\begin{proposition}[Two-point optimality with nonzero ISI at small peak intensity]
\label{prop:small_peak_isi_binary}
Fix $\lambda_0>0$ and a finite tap profile with $L\ge1$, $h_0>0$,
$h_\ell\ge0$, $\sum_{\ell=0}^{L}h_\ell=1$, and
$\sum_{\ell=1}^{L}h_\ell>0$.
There exists $A_{\rm sp}=A_{\rm sp}(\lambda_0,\mathbf h)>0$ such that,
for every $0<A<A_{\rm sp}$ and every $P_{\rm avg}>0$, each
maximizer $F^\star$ of \eqref{eq:capacity_def} satisfies
\begin{equation}
\label{eq:small_peak_binary_support}
\operatorname{supp}(F^\star)=\{0,A\},
\qquad |\operatorname{supp}(F^\star)|=2.
\end{equation}
Consequently, in this genuine-ISI regime the optimization reduces
exactly to the one-parameter family
\begin{equation}
\label{eq:small_peak_binary_family}
F_p=(1-p)\delta_0+p\delta_A,
\qquad 0\le p\le\min\{1,P_{\rm avg}/A\}.
\end{equation}
\end{proposition}
}

\section{Proof of Theorem~\ref{thm:finite_support}: Finite-Support Structure}
\label{sec:proof_finite_support}

We now prove Theorem~\ref{thm:finite_support}. The argument starts from the HMM representation induced by a fixed i.i.d. input law. We first prove a filtering-forgetting estimate for the predictive filters and extend it to finite signed perturbations. This estimate gives the continuity needed for the existence of an optimizer and also justifies the entropy-rate first variation. After obtaining the holomorphic extension of the influence function, we derive the KKT condition for the optimal marginal law. The proof is completed by showing that the holomorphic KKT defect cannot vanish identically, because its value on the positive real axis grows faster than any linear function.

\subsection{Uniform Filtering-Forgetting Estimate}

As noted in Section~\ref{subsec:HMM}, for a fixed input law $F$, the finite-memory Poisson channel can be represented as an HMM whose hidden state is the input window $Q_t=(X_{t-L},\ldots,X_t)$. We write $\mathsf Q\triangleq[0,A]^{L+1}$ for the hidden-state space, so $Q_t\in\mathsf Q$, and the observation is the photon count $Y_t$. Given an initial law $\mu$ for the hidden state, let $\mathsf P_F^\mu$ denote the probability law of the HMM when $Q_0\sim\mu$. The predictive filter at time $t$ is the conditional distribution of $Q_t$ given the past observations:
$ \Pi_t^\mu(E) \triangleq \mathsf P_F^\mu\bigl(Q_t\in E \mid Y_1,\ldots,Y_{t-1}\bigr)$ for $E\in\mathcal B(\mathsf Q)$. Thus $\Pi_t^\mu$ is the one-step-ahead belief state used to predict $Y_t$. We also write $\widehat\Pi_t^\mu(E)\triangleq\mathsf P_F^\mu(Q_t\in E\mid Y_1^t)$ for the posterior filter.

For $q=(q_{-L},\ldots,q_0)\in\mathsf Q$, define $\lambda(q)\triangleq \lambda_0+\sum_{\ell=0}^{L}h_\ell q_{-\ell}$. If $\pi$ is a probability measure on $\mathsf Q$, the one-step predictive probability mass function is
\begin{equation}
\label{eq:predictive_pmf_def}
\begin{aligned}
 p_\pi(y)
 &\triangleq
 \int_{\mathsf Q}
 \exp[-\lambda(q)]
 \frac{\lambda(q)^y}{y!}\,d\pi(q),
 \qquad y\in\mathbb N_0 .
\end{aligned}
\end{equation}
The infinite-history entropy rate $h_F(\mathbf Y)=H_F(Y_0|Y_{-\infty}^{-1})$ is determined by the stationary predictive filter generated from the remote observation past. We therefore need a stability estimate showing that two filters with different remote initializations produce asymptotically identical one-step predictive laws after a long common observation history.

The norm $\lVert\cdot\rVert_{\rm TV}$ denotes the total-variation norm for finite signed measures; the convention is recalled in Appendix~\ref{app:measure_tools}. The proof below works with the unnormalized block kernel and applies the Hilbert projective metric and Birkhoff contraction in Appendix~\ref{app:hilbert_metric}. No atom assumption is imposed on $F$. The common dominating terminal measure is $F^{\otimes N}$ itself, and the constants are obtained only from likelihood-ratio ranges and good-block probabilities.

\begin{lemma}
\label{lem:uniform_filter_forgetting}
Fix $A<\infty$, $\lambda_0>0$, and nonnegative taps $h_0,\ldots,h_L$ with $\sum_{\ell=0}^L h_\ell=1$. For any $F\in\mathcal P([0,A])$, consider the hidden state $Q_t\triangleq (X_{t-L},\ldots,X_t)\in\mathsf Q$. Let $\Pi_t^\mu$ and $\Pi_t^\nu$ be two predictive filters at time $t$, driven by the same observation sequence $Y_1,\ldots,Y_{t-1}$ but initialized from arbitrary probability laws $\mu$ and $\nu$ on $\mathsf Q$. Let $\theta$ denote an arbitrary hidden-state initial law that generates the common observation sequence, and take expectations under $\mathsf P_F^\theta$. Then there exist constants $C<\infty$ and $0<\beta<1$, depending only on $(\lambda_0,A,\mathbf h,L)$ and independent of $F,\mu,\nu,$ and $\theta$, such that
\begin{equation}
\label{eq:uff_tv_forgetting}
\begin{aligned}
 \mathbb E_F^\theta
 \bigl[
 \lVert\Pi_t^\mu-\Pi_t^\nu\rVert_{\rm TV}
 \bigr]
 &\le
 C\beta^t\lVert\mu-\nu\rVert_{\rm TV} .
\end{aligned}
\end{equation}
Moreover, the associated one-step predictive log-probabilities satisfy
\begin{equation}
\label{eq:uff_log_forgetting}
\begin{aligned}
 \mathbb E_F^\theta
 \left[
 \left|\log p_{\Pi_t^\mu}(Y_t)-\log p_{\Pi_t^\nu}(Y_t)\right|
 \right]
 &\le
 C\beta^t\lVert\mu-\nu\rVert_{\rm TV} .
\end{aligned}
\end{equation}
Consequently,
\begin{equation}
\label{eq:uff_entropy_history}
\begin{aligned}
&\sup_{F\in\mathcal P([0,A])}
\mathbb E_F
\bigl[
\bigl|
\log P_F(Y_0|Y_{-m}^{-1})
-
\log P_F(Y_0|Y_{-\infty}^{-1})
\bigr|
\bigr] \le C\beta^m .
\end{aligned}
\end{equation}
\end{lemma}

\begin{proof}
For $q\in\mathsf Q$, the Poisson mean satisfies
\begin{equation}
\label{eq:uff_lambda_bound}
\lambda_0\le \lambda(q)\le B_A .
\end{equation}
Let $g_y(q)\triangleq \exp[-\lambda(q)]\lambda(q)^y/y!$. For any $q,q'\in\mathsf Q$,
\begin{equation}
\label{eq:one_symbol_likelihood_ratio}
\frac{g_y(q)}{g_y(q')}
\le
\exp(B_A-\lambda_0)
\left(\frac{B_A}{\lambda_0}\right)^y .
\end{equation}
Group the hidden chain into blocks of length $N=L+1$. Starting from an initial hidden state $q$, let $u_1^N=(u_1,\ldots,u_N)$ denote the fresh input symbols drawn during the next $N$ updates, and let $Q_i(q,u_1^i)$ be the hidden state at time $i$ generated from $q$ and $u_1,\ldots,u_i$. After $N=L+1$ updates, the hidden state consists entirely of the fresh symbols $u_1^N$. Thus the terminal state of the block may be identified with $u_1^N$, and the common reference measure for the terminal state is $F^{\otimes N}$.

For a fixed output block $y_1^N$, define
\begin{equation}
\label{eq:block_likelihood}
G_{y_1^N}(q,u_1^N)
\triangleq
\prod_{i=1}^{N}
 g_{y_i}\bigl(Q_i(q,u_1^i)\bigr).
\end{equation}
The block positive kernel is
\begin{equation}
\label{eq:block_positive_kernel}
\mathcal L_{y_1^N}(q,du_1^N)
\triangleq
G_{y_1^N}(q,u_1^N)F^{\otimes N}(du_1^N).
\end{equation}
For two entrance states $q,q'$ and the same fresh-input vector $u_1^N$, \eqref{eq:one_symbol_likelihood_ratio} gives
\begin{equation}
\label{eq:block_likelihood_ratio}
\begin{aligned}
\frac{G_{y_1^N}(q,u_1^N)}
     {G_{y_1^N}(q',u_1^N)}
&\le
\exp\bigl(N(B_A-\lambda_0)\bigr)
\left(\frac{B_A}{\lambda_0}\right)^{\sum_{i=1}^N y_i} \triangleq R(y_1^N),
\end{aligned}
\end{equation}
and the same bound holds with $q$ and $q'$ interchanged.

For an entrance probability measure $\pi$ on $\mathsf Q$, define the normalized block Bayes operator
\begin{equation}
\label{eq:block_bayes_operator}
\begin{aligned}
\Phi_{y_1^N}(\pi)(E)
&\triangleq
\frac{\displaystyle
 \int_{\mathsf Q}\int_E
 G_{y_1^N}(q,u_1^N)
 F^{\otimes N}(du_1^N)d\pi(q)}
{\displaystyle
 \int_{\mathsf Q}\int
 G_{y_1^N}(q,v_1^N)
 F^{\otimes N}(dv_1^N)d\pi(q)} ,
\end{aligned}
\end{equation}
where $E\in\mathcal B([0,A]^N)$.

{
Fix an arbitrary reference state $q^\circ\in\mathsf Q$ and define the finite positive measure
\begin{equation}
\Lambda_{y_1^N}(E)\triangleq\mathcal L_{y_1^N}(q^\circ,E).
\end{equation}
By \eqref{eq:block_likelihood_ratio} and the same bound with $q$ and $q^\circ$ interchanged,
\begin{equation}
\label{eq:hilbert_mixing_kernel}
R(y_1^N)^{-1}\Lambda_{y_1^N}(E)
\le
\mathcal L_{y_1^N}(q,E)
\le
R(y_1^N)\Lambda_{y_1^N}(E)
\end{equation}
for every $q$ and measurable $E$. Thus, for each realized output block, $\mathcal L_{y_1^N}$ is an $\varepsilon(y_1^N)$-mixing positive kernel with
\begin{equation}
\varepsilon(y_1^N)\triangleq R(y_1^N)^{-1}.
\end{equation}
The reference measure may depend on $F$ and on the realized block, but the mixing constant is controlled only by $(\lambda_0,A,\mathbf h,L)$ and the observed counts. By the Hilbert-metric facts in Appendix~\ref{app:hilbert_metric}~\cite{legland2004hilbert}, the Birkhoff contraction coefficient of this block kernel satisfies
\begin{equation}
\label{eq:birkhoff_block_coeff}
\tau(y_1^N)
\le
\frac{1-\varepsilon(y_1^N)^2}{1+\varepsilon(y_1^N)^2}
=
\frac{R(y_1^N)^2-1}{R(y_1^N)^2+1}<1.
\end{equation}
}

Fix once and for all $M=0$ and define the good block event
\begin{equation}
\label{eq:good_block_event}
\mathcal G_M\triangleq\{y_1\le M,\ldots,y_N\le M\}.
\end{equation}
On $\mathcal G_M$,
\begin{equation}
\label{eq:good_block_RM}
R(y_1^N)\le R_M
\triangleq
\exp\bigl(N(B_A-\lambda_0)\bigr)
\left(\frac{B_A}{\lambda_0}\right)^{NM}.
\end{equation}
{
On $\mathcal G_M$, \eqref{eq:birkhoff_block_coeff} gives the deterministic contraction bound
\begin{equation}
\label{eq:good_block_rho}
\tau(y_1^N)\le\rho_M
\triangleq
\frac{R_M^2-1}{R_M^2+1}<1.
\end{equation}
For arbitrary blocks we use only $\tau(y_1^N)\le1$.
}

The good-block probability is uniformly positive. For $\lambda\in[\lambda_0,B_A]$, define
\begin{equation}
\label{eq:good_block_cm}
c_M\triangleq
\inf_{\lambda\in[\lambda_0,B_A]}
\sum_{y=0}^{M}e^{-\lambda}\frac{\lambda^y}{y!}>0.
\end{equation}
The strict positivity of $c_M$ follows because the function
$\sum_{y=0}^{M}e^{-\lambda}\lambda^y/y!$ is continuous
and strictly positive on the compact interval $[\lambda_0,B_A]$. Conditional on $Q_0=q$ and on any realization of $u_1^N$, all block-output means lie in $[\lambda_0,B_A]$, hence
\begin{equation}
\label{eq:good_block_cond_lower}
\mathsf P(\mathcal G_M|Q_0=q,u_1^N)\ge c_M^N.
\end{equation}
Averaging over $F^{\otimes N}$ gives
\begin{equation}
\label{eq:good_block_probability}
p_M\triangleq\inf_{q\in\mathsf Q}
\mathsf P(\mathcal G_M|Q_0=q)
\ge c_M^N>0.
\end{equation}
Thus, for every predictive entrance filter $\pi$,
\begin{equation}
\label{eq:good_block_filter_probability}
\mathsf P_\pi(\mathcal G_M)
=
\int_{\mathsf Q}\mathsf P(\mathcal G_M|Q_0=q)d\pi(q)
\ge p_M .
\end{equation}
{
Let $Y^{(k)}\triangleq(Y_{(k-1)N+1},\ldots,Y_{kN})$ be the $k$-th complete observation block. For $j\ge1$, the deterministic product estimate in Appendix~\ref{app:hilbert_metric} gives
\begin{equation}
\label{eq:hilbert_pathwise_complete_blocks}
\begin{aligned}
\bigl\|\widehat\Pi_{jN}^\mu-\widehat\Pi_{jN}^\nu\bigr\|_{\rm TV}
\le
\frac{2}{\log3}R(Y^{(1)})^2
\prod_{k=2}^{j}\tau(Y^{(k)})
\lVert\mu-\nu\rVert_{\rm TV}.
\end{aligned}
\end{equation}
The first block regularizes arbitrary initial laws; subsequent blocks contract in the projective metric. This is precisely where Bayes normalization is handled: the Hilbert metric is invariant under positive rescaling, so no coupling of pre-observation entrance laws is used after conditioning.

Let $\mathcal F_{k-1}$ be the observation sigma-field up to the beginning of block $k$. By \eqref{eq:good_block_filter_probability} and \eqref{eq:good_block_rho},
\begin{equation}
\label{eq:conditional_tau_bound}
\mathbb E_F^\theta[\tau(Y^{(k)})\mid\mathcal F_{k-1}]
\le
1-p_M(1-\rho_M)
\triangleq\beta_0<1.
\end{equation}
The random regularization factor has a uniform moment. Indeed, with $a=(B_A/\lambda_0)^2$, conditional on the state and input trajectory in the first block, the block outputs are conditionally independent and
\begin{equation}
\mathbb E\left[a^{\sum_{i=1}^{N}Y_i}\middle|Z_1^N\right]
=\exp\left\{(a-1)\sum_{i=1}^{N}Z_i\right\}
\le \exp\{NB_A(a-1)\},
\end{equation}
so that
\begin{equation}
\label{eq:R_first_block_moment}
\sup_{F,\theta}\mathbb E_F^\theta[R(Y^{(1)})^2]
\le C_R<\infty,
\end{equation}
where $C_R$ depends only on $(\lambda_0,A,L)$. Iterating \eqref{eq:conditional_tau_bound} backward from the last block therefore yields
\begin{equation}
\label{eq:complete_block_filter_forgetting}
\mathbb E_F^\theta
\bigl[\lVert\widehat\Pi_{jN}^\mu-\widehat\Pi_{jN}^\nu\rVert_{\rm TV}\bigr]
\le C_1\beta_0^{j-1}\lVert\mu-\nu\rVert_{\rm TV}.
\end{equation}

It remains to pass from complete blocks to arbitrary times without assuming that normalized Bayes updates are nonexpansive. For $0\le r<N$, let $\mathcal B_{y_1^r}$ denote the unnormalized kernel of an $r$-symbol fragment and let $b_{y_1^r}(q)=\mathcal B_{y_1^r}(q,\mathsf Q)$. The same likelihood-ratio argument as in \eqref{eq:block_likelihood_ratio} gives
\begin{equation}
\frac{\sup_q b_{y_1^r}(q)}{\inf_q b_{y_1^r}(q)}
\le
R_r(y_1^r)
\triangleq
\exp(r(B_A-\lambda_0))
\left(\frac{B_A}{\lambda_0}\right)^{\sum_{i=1}^r y_i}.
\end{equation}
The normalization inequality in Appendix~\ref{app:hilbert_metric} then gives
\begin{equation}
\label{eq:fragment_lipschitz}
\left\|
\frac{\mu\mathcal B_{y_1^r}}{(\mu\mathcal B_{y_1^r})(\mathsf Q)}
-
\frac{\nu\mathcal B_{y_1^r}}{(\nu\mathcal B_{y_1^r})(\mathsf Q)}
\right\|_{\rm TV}
\le 2R_r(y_1^r)\lVert\mu-\nu\rVert_{\rm TV}.
\end{equation}
All exponential moments of the Poisson counts are uniformly finite because their means lie in $[\lambda_0,B_A]$; hence, since $r<N$,
\begin{equation}
\label{eq:fragment_expected_lipschitz}
\sup_{F,\theta}\mathbb E_F^\theta[R_r(Y_1^r)]\le C_{\rm fr}<\infty.
\end{equation}
Write $n=r+jN$ with $0\le r<N$. Condition first on the initial fragment and then apply \eqref{eq:complete_block_filter_forgetting} to the following complete blocks. The constants are uniform in the entrance and true state laws, and \eqref{eq:fragment_expected_lipschitz} gives
\begin{equation}
\mathbb E_F^\theta
\bigl[\lVert\widehat\Pi_n^\mu-\widehat\Pi_n^\nu\rVert_{\rm TV}\bigr]
\le C_2\beta_0^{j-1}\lVert\mu-\nu\rVert_{\rm TV}.
\end{equation}
With $\beta=\beta_0^{1/N}$, the finite remainder and the finitely many cases $j=0$ are absorbed into the constant, yielding $C_3\beta^n\lVert\mu-\nu\rVert_{\rm TV}$. Finally, $\Pi_t^\mu$ is obtained from $\widehat\Pi_{t-1}^\mu$ by the ordinary state-transition kernel, which is nonexpansive in total variation. This proves \eqref{eq:uff_tv_forgetting}.
}

We next pass from filter distance to log-predictive distance. From \eqref{eq:uff_lambda_bound}, for every hidden state $q$,
\begin{equation}
\label{eq:emission_bounds_uniform}
 e^{-B_A}\frac{\lambda_0^y}{y!}
 \le
 g_y(q)
 \le
 e^{-\lambda_0}\frac{B_A^y}{y!} .
\end{equation}
Integrating with respect to any probability filter $\pi$ gives
\begin{equation}
\label{eq:uff_predictive_bounds}
 e^{-B_A}\frac{\lambda_0^y}{y!}
 \le
 p_\pi(y)
 \le
 e^{-\lambda_0}\frac{B_A^y}{y!} .
\end{equation}
Moreover, for two probability filters $\pi$ and $\pi'$,
\begin{equation}
\label{eq:uff_predictive_tv_lip}
\begin{aligned}
 |p_\pi(y)-p_{\pi'}(y)|
 &\le
 e^{-\lambda_0}\frac{B_A^y}{y!}
 \lVert\pi-\pi'\rVert_{\rm TV} .
\end{aligned}
\end{equation}
Since $|\log a-\log b|\le |a-b|/\min\{a,b\}$ for $a,b>0$, \eqref{eq:uff_predictive_bounds} and \eqref{eq:uff_predictive_tv_lip} imply
\begin{equation}
\label{eq:uff_log_lip}
\begin{aligned}
 |\log p_\pi(y)-\log p_{\pi'}(y)|
 &\le
 e^{B_A-\lambda_0}
 \left(\frac{B_A}{\lambda_0}\right)^y
 \lVert\pi-\pi'\rVert_{\rm TV} .
\end{aligned}
\end{equation}
Conditionally on the past observations under $\mathsf P_F^\theta$, $Y_t$ has a predictive distribution $p_\pi$ for some probability filter $\pi$. Hence
\begin{equation}
\label{eq:uniform_exp_moment_predictive}
\begin{aligned}
 \sup_{\pi}
 \sum_{y=0}^{\infty}
 \left(\frac{B_A}{\lambda_0}\right)^y p_\pi(y)
 &\le
 e^{-\lambda_0}
 \sum_{y=0}^{\infty}
 \frac{(B_A^2/\lambda_0)^y}{y!}<\infty .
\end{aligned}
\end{equation}
Using \eqref{eq:uff_log_lip} and \eqref{eq:uniform_exp_moment_predictive} conditionally on $Y_1^{t-1}$ gives
\begin{equation}
\label{eq:conditional_log_lip_filter}
\begin{aligned}
&\mathbb E_F^\theta\left[
\left|
\log p_{\Pi_t^\mu}(Y_t)-\log p_{\Pi_t^\nu}(Y_t)
\right|
\,\middle|\,Y_1^{t-1}
\right] \le
C_0\lVert\Pi_t^\mu-\Pi_t^\nu\rVert_{\rm TV},
\end{aligned}
\end{equation}
where $C_0<\infty$ is independent of $F$. The constant in \eqref{eq:conditional_log_lip_filter} depends only on $\lambda_0$ and $B_A$, because the exponential moment in \eqref{eq:uniform_exp_moment_predictive} is taken uniformly over all predictive filters $\pi$ and all input laws supported on $[0,A]$. Taking expectation and using \eqref{eq:uff_tv_forgetting} proves \eqref{eq:uff_log_forgetting}. Finally, under the stationary law, shift the time origin to the left end of the common observation segment $Y_{-m}^{-1}$. The finite-history and infinite-history predictors are two filters driven by this same segment but initialized by two probability measures whose total-variation distance is at most $2$. Since \eqref{eq:uff_log_forgetting} is uniform in $\mu$, $\nu$, and the true law $\theta$, it may be applied conditionally on the remote past and then averaged, which gives \eqref{eq:uff_entropy_history} after absorbing the factor $2$ into $C$.

\end{proof}

{
\begin{corollary}
\label{cor:pathwise_block_contraction}
For deterministic observation blocks $y^{(1)},\ldots,y^{(j)}$ and any entrance probability measures $\pi,\pi'$ on $\mathsf Q$,
\begin{equation}
\label{eq:pathwise_block_contraction}
\begin{aligned}
&\left\|
\Phi_{y^{(j)}}\circ\cdots\circ\Phi_{y^{(1)}}(\pi)
-
\Phi_{y^{(j)}}\circ\cdots\circ\Phi_{y^{(1)}}(\pi')
\right\|_{\rm TV}\\
&\qquad\le
\frac{2}{\log3}R(y^{(1)})^2
\prod_{k=2}^{j}
\frac{R(y^{(k)})^2-1}{R(y^{(k)})^2+1}
\lVert\pi-\pi'\rVert_{\rm TV}.
\end{aligned}
\end{equation}
In particular, if $N_M(y^{(2:j)})$ is the number of good blocks among blocks $2,\ldots,j$, then the product in \eqref{eq:pathwise_block_contraction} is at most $\rho_M^{N_M(y^{(2:j)})}$. The first-block factor is the regularization cost for arbitrary initial laws; strict projective contraction is accumulated by the subsequent blocks.
\end{corollary}

\begin{proof}
Apply \eqref{eq:hilbert_mixing_kernel} to the first block and the Birkhoff contraction inequality to each subsequent block, followed by the Hilbert-to-total-variation inequality in Appendix~\ref{app:hilbert_metric}. The good-block statement follows from \eqref{eq:good_block_rho}.
\end{proof}
}

\begin{lemma}
\label{lem:signed_initial_extension}
Let $s\ge0$ be fixed and $F\in\mathcal P([0,A])$. Consider a finite signed or complex signed measure $\nu$ on an augmented history space $\mathcal W_s$. An element $w\in\mathcal W_s$ contains (i) the physical hidden state at time $s$, (ii) a posterior baseline filtering state $\widehat\pi_s\in\mathcal P(\mathsf Q)$ that has already incorporated the observation $Y_s$, and (iii) any finite observation coordinates needed by the functional under consideration. We equip $\mathcal P(\mathsf Q)$ with its weak-Borel $\sigma$-field and $\mathcal W_s$ with the corresponding product $\sigma$-field. Only the outer measure $\nu$ may be signed or complex; the baseline filter coordinate is always a genuine real probability measure. Assume that $\nu(\mathcal W_s)=0$ and $\lVert\nu\rVert_{\rm TV}<\infty$. From time $s+1$ onward, let all hidden transitions and observation kernels be the ordinary real kernels induced by the input law $F$, and propagate $\widehat\pi_s$ by the ordinary prediction/update recursion; $K_{s+1:t}$ denotes this induced augmented-state kernel. Let $\nu K_{s+1:t}$ denote the resulting signed measure on the future variables up to time $t$. Then there exist constants $C<\infty$ and $0<\beta<1$, depending only on $(\lambda_0,A,\mathbf h,L)$ and not on $F$, $s$, or $\nu$, such that, for every $t>s$,
\begin{equation}
\label{eq:signed_history_functional_bound}
\begin{aligned}
&\biggl|
\bigl\langle
-\log P_F(Y_t|Y_{<t}),\nu K_{s+1:t}
\bigr\rangle
\biggr|
\le
C\beta^{t-s}\lVert\nu\rVert_{\rm TV}.
\end{aligned}
\end{equation}
The same bound holds for the real and imaginary parts separately when $\nu$ is complex signed, after enlarging $C$ if necessary. The estimate also holds, with the same constants, when $-\log P_F(Y_t|Y_{<t})$ is replaced by the finite-history information density $-\log P_F(Y_t|Y_{t-m}^{t-1})$, uniformly over all $m\ge0$.
\end{lemma}

\begin{proof}
We give the proof for real signed measures; the complex case follows by applying the same argument to the real and imaginary parts. Let $\nu=\nu^+-\nu^-$ be the Jordan decomposition. Since $\nu(\mathcal W_s)=0$, both positive parts have the same mass $m=\nu^+(\mathcal W_s)=\nu^-(\mathcal W_s)=\lVert\nu\rVert_{\rm TV}/2$. If $m=0$, there is nothing to prove. Otherwise set $\bar\nu^+=\nu^+/m$ and $\bar\nu^-=\nu^-/m$.

For an augmented history $w\in\mathcal W_s$, let $R_t(w)$ be the expectation of the baseline information density $-\log P_F(Y_t|Y_{<t})$ when the future after time $s$ is generated by the ordinary real HMM kernels starting from the physical hidden state and posterior baseline filtering state encoded by $w$. Since $\nu$ has zero total mass, for any fixed reference history $w_0$,
\begin{equation}
\label{eq:signed_history_centering}
\begin{aligned}
&\bigl\langle
-\log P_F(Y_t|Y_{<t}),\nu K_{s+1:t}
\bigr\rangle=
\int_{\mathcal W_s}
\bigl(R_t(w)-R_t(w_0)\bigr)d\nu(w).
\end{aligned}
\end{equation}

{
It remains to bound $|R_t(w)-R_t(w_0)|$ uniformly over finite histories. Put $N=L+1$ and couple the two future input processes after time $s$ by using the same fresh inputs. The observations during the first $N-1$ steps are allowed to differ, because their Poisson means may depend on the two different entrance states. After $N$ state updates, however, the two hidden-state windows are identical, since both consist entirely of the same fresh symbols. From that time onward use the same fresh inputs and the same Poisson randomness, so the subsequent observation string is common to the two constructions.

Conditional on the coupled variables up to time $s+N$, the two baseline predictors are therefore two filters with possibly different posterior baseline filtering states, driven thereafter by the same observations under a common true hidden state. Their two filter initializations are probability measures and hence have total-variation distance at most $2$. For $t-s>N$, applying \eqref{eq:uff_log_forgetting}, with the arbitrary true initialization allowed in Lemma~\ref{lem:uniform_filter_forgetting}, gives
\begin{equation}
\label{eq:finite_history_Rt_difference}
\lvert R_t(w)-R_t(w_0)\rvert
\le
C\beta^{t-s-N}
\le C'\beta^{t-s}
\end{equation}
uniformly over $w,w_0$. The finitely many times $t-s\le N$ are absorbed into the constant by the uniform Poisson logarithmic moment bound \eqref{eq:uniform_exp_moment_predictive}. For the finite-history version, let $R_{t,m}(w)$ denote the expected finite-history information density and put $q=s+N$. If $t-m>q$, the two coupled constructions have identical observations throughout $Y_{t-m}^{t}$, so $R_{t,m}(w)=R_{t,m}(w_0)$. If $t-m\le q<t$, let $\pi_q$ and $\pi_q'$ be the two real baseline filtering distributions at the end of the possibly different portions of the finite conditioning windows. They satisfy $\lVert\pi_q-\pi_q'\rVert_{\rm TV}\le2$. Conditional on the coupled histories through $q$, the physical hidden state and all subsequent observations are common; applying \eqref{eq:uff_log_forgetting} from $q$ onward under this common true observation law gives
\[
\left|R_{t,m}(w)-R_{t,m}(w_0)\right|
\le 2C\beta^{t-q}
\le C'\beta^{t-s}.
\]
For $t\le q$, the uniform Poisson log-moment bound absorbs these finitely many lags into $C'$. In all cases $C'$ is independent of $m$, $F$, $w$, and $w_0$. Replacing $R_t$ by $R_{t,m}$ in the centering identity \eqref{eq:signed_history_centering} proves the finite-history claim. Combining \eqref{eq:signed_history_centering} and \eqref{eq:finite_history_Rt_difference} yields \eqref{eq:signed_history_functional_bound}.
}
\end{proof}
\subsection{Existence of an Optimizer}

We next show that the supremum in the i.i.d.-constrained mutual-information-rate problem is attained.

\begin{lemma}
\label{lem:existence}
The optimization problem
$ \sup_{F\in\mathcal P_A(P_{\rm avg})}\mathcal I(F) $
admits at least one maximizer $F^\star\in\mathcal P_A(P_{\rm avg})$.
\end{lemma}

\begin{proof}
Since $[0,A]$ is compact, the space $\mathcal P([0,A])$ is weakly compact~\cite{billingsley1999convergence}. The average-power constraint is weakly closed because $x$ is a bounded continuous function on $[0,A]$. Hence $\mathcal P_A(P_{\rm avg})$ is weakly compact.

It remains to prove that $\mathcal I(F)$ is continuous on $\mathcal P_A(P_{\rm avg})$. Under the peak constraint and the tap normalization,
\begin{equation}
 \lambda_0 \le Z_k \le \lambda_0+A .
\end{equation}
Therefore, all finite-dimensional output distributions have uniformly dominated Poisson tails. For every finite block length $n$, the block mutual information $I_F(X_1^n;Y_1^n)$ is continuous in $F$ under weak convergence by dominated convergence.

It remains to pass from finite blocks to the entropy rate uniformly in $F$. For $m\ge1$, set
\begin{equation}
 h_m(F) \triangleq H_F(Y_0|Y_{-m}^{-1}).
\end{equation}
For each fixed $m$, $h_m(F)$ is continuous in $F$ by the same finite-dimensional domination argument. {
To compare $h_m(F)$ with the infinite-history entropy rate, use \eqref{eq:uff_entropy_history} directly:
\begin{equation}
\begin{aligned}
|h_m(F)-h_F(\mathbf Y)|
&\le
\mathbb E_F
\left|
\log P_F(Y_0|Y_{-m}^{-1})
-
\log P_F(Y_0|Y_{-\infty}^{-1})
\right|\\
&\le C\beta^m.
\end{aligned}
\end{equation}
Hence
\begin{equation}
 \sup_{F\in\mathcal P_A(P_{\rm avg})}
 |h_m(F)-h_F(\mathbf Y)|
 \le C\beta^m .
\end{equation}
The constants are uniform in $F$ by Lemma~\ref{lem:uniform_filter_forgetting}.
}
Hence $h_F(\mathbf Y)$ is the uniform limit of the continuous functions $h_m(F)$, and is continuous. The conditional entropy term $\mathbb E_F[\phi(Z_0)]$ is continuous because $\phi$ is continuous on $[\lambda_0,\lambda_0+A]$. Thus $\mathcal I(F)$ is continuous on the compact admissible set, and the supremum is attained.
\end{proof}

\subsection{First Variations and Entropy-Rate Influence Function}

For $x\in[0,A]$, the point-mass influence is a directional derivative of the information-rate functional on the ambient probability simplex $\mathcal P([0,A])$; the direction $\delta_x-F$ need not itself satisfy the average-power constraint. Fixing one symbol produces a direct change in a finite output window and a later filter-mediated effect. Its holomorphic extension will be denoted by a different symbol, $\Psi_F$, because values on the positive real axis outside $[0,A]$ are analytic-continuation values rather than feasible point-mass perturbations.

For a complex number $z\in\Omega$, the notation involving a fixed symbol $X_j=z$ is defined as follows. In any finite window affected by $X_j$, each Poisson kernel whose mean contains $z$ is replaced by its analytic expression $e^{-\zeta(z)}\zeta(z)^y/y!$. This produces a finite complex measure on the corresponding output block. The logarithmic factor against which this complex measure is paired is always the baseline information density generated by the unperturbed law $F$; the baseline predictive filter is not complexified. When $z\in[0,A]$, the same expression coincides with the ordinary probabilistic expectation under the experiment in which $X_j$ is fixed to $z$.

More explicitly, for a finite history length $m$, we define
\begin{equation}
\label{eq:finite_history_baseline_density}
\begin{aligned}
\ell_{F,m}(y_t|y_{t-m}^{t-1})
\triangleq
-\log P_F(Y_t=y_t|Y_{t-m}^{t-1}=y_{t-m}^{t-1}),
\end{aligned}
\end{equation}
with the convention that the conditioning is interpreted under the stationary baseline law $F$. For each $t,m$ and $z\in\Omega$, let $\Lambda_{t,m}^{(z)}$ be the finite-dimensional complex measure obtained by fixing $X_0=z$ in the affected Poisson kernels and keeping all other input symbols distributed according to $F$. The finite-history perturbation functional is
\begin{equation}
\label{eq:Delta_tm_definition}
\begin{aligned}
\Delta_{t,m}(z)\triangleq
\left\langle
\ell_{F,m}(Y_t|Y_{t-m}^{t-1}),
\Lambda_{t,m}^{(z)}-\Lambda_{t,m}^{(F)}
\right\rangle .
\end{aligned}
\end{equation}
Here $\Lambda_{t,m}^{(F)}$ denotes the corresponding baseline finite-dimensional law. The signed or complex signed measure in \eqref{eq:Delta_tm_definition} has total mass zero. For real $x\in[0,A]$, \eqref{eq:Delta_tm_definition} is the usual expectation difference produced by the point-mass perturbation $\delta_x-F$; for nonreal $z$, it is only a finite-dimensional pairing.

\begin{lemma}
\label{lem:finite_block_derivative}
Fix $F\in\mathcal P([0,A])$. Let $h_+\triangleq \max_{0\le r\le L:h_r>0}h_r$ and $\Omega\triangleq
\{z\in\mathbb C:\operatorname{Re}z>-\lambda_0/h_+\}$. For every compact set $K\Subset\Omega$, define the finite-block analytic influence
\begin{equation}
\label{eq:finite_block_derivative_function}
D_{n,F}(z)
\triangleq
-\frac1n\sum_{j=1-L}^{n}\sum_{y^n}
\bigl[P_{F,n}^{(j,z)}(y^n)-P_{F,n}(y^n)\bigr]\log P_{F,n}(y^n),
\end{equation}
where $P_{F,n}^{(j,z)}$ is the finite-dimensional analytic continuation obtained by fixing $X_j=z$. For real $z\in[0,A]$, $D_{n,F}(z)$ is the genuine normalized directional derivative of $H_F(Y_1^n)$ in the direction $\delta_z-F$. These functions converge locally uniformly on $K$ to
\begin{equation}
\label{eq:entropy_rate_derivative_series}
\dot h_{Y,F}^{\rm hol}(z)
\triangleq
\sum_{t=0}^{\infty}\Delta_t(z),
\end{equation}
where $\Delta_t(z)$ is the locally uniform limit of the finite-history analytic functions $\Delta_{t,m}(z)$ defined in \eqref{eq:Delta_tm_definition}. Equivalently, $\Delta_t(z)$ is the analytic continuation of the change in the baseline information density $\ell_F(Y_t|Y_{<t})=-\log P_F(Y_t|Y_{<t})$ caused by fixing $X_0=z$, with the baseline log-likelihood kept fixed and only the finite-dimensional measure complexified. Moreover,
\begin{equation}
\label{eq:delta_series_uniform_bound_new}
\sum_{t=0}^{\infty}\sup_{z\in K}|\Delta_t(z)|<\infty .
\end{equation}
All bounds can be chosen uniformly over $F\in\mathcal P([0,A])$, and in particular
\begin{equation}
\label{eq:finite_block_uniform_in_F}
\sup_{F\in\mathcal P([0,A])}\sup_{z\in K}
\left|D_{n,F}(z)-\dot h_{Y,F}^{\rm hol}(z)\right|\longrightarrow0.
\end{equation}
For real $x\in[0,A]$, $\dot h_{Y,F}^{\rm hol}(x)$ equals the directional derivative of the output entropy rate in the direction $\delta_x-F$.
\end{lemma}

\begin{proof}
For a finite block, let $P_{F,n}$ be the output distribution of $Y_1^n$ under the stationary input law $F$. Let $P_{F,n}^{(j,z)}$ be the finite-dimensional analytic continuation obtained by fixing the single input symbol $X_j=z$ and keeping all other input symbols distributed according to $F$. The complex Poisson series is absolutely convergent and sums to one, so $\sum_{y^n}P_{F,n}^{(j,z)}(y^n)=1$ also for nonreal $z\in\Omega$. Hence \eqref{eq:finite_block_derivative_function} is a well-defined analytic continuation. For real $z\in[0,A]$, differentiating the product input measure in the direction $\delta_z-F$ gives
\begin{equation}
\label{eq:block_entropy_derivative_point_mass}
\begin{aligned}
nD_{n,F}(z) =
-\sum_{j=1-L}^{n}\sum_{y^n}
\bigl[P_{F,n}^{(j,z)}(y^n)-P_{F,n}(y^n)\bigr]
\log P_{F,n}(y^n).
\end{aligned}
\end{equation}
The derivative of the normalization term in the entropy formula vanishes because $\sum_{y^n}[P_{F,n}^{(j,z)}(y^n)-P_{F,n}(y^n)]=0$. By the chain rule,
\begin{equation}
\label{eq:block_chain_rule_info_density}
-
\log P_{F,n}(Y_1^n)
=
\sum_{s=1}^{n}-\log P_F(Y_s|Y_1^{s-1}).
\end{equation}
{
For a fixed pair $(j,s)$, the contribution vanishes when $j>s$, because the causal outputs $Y_1^s$ do not depend on the future input $X_j$. For $1-L\le j\le s$, put $t=s-j$. After shifting time by $j$, stationarity identifies the corresponding contribution exactly with $\Delta_{t,s-1}(z)$: the conditioning window $Y_1^{s-1}$ becomes $Y_{t-(s-1)}^{t-1}$. Therefore no separate boundary correction is required, and we have the exact triangular identity
\begin{equation}
\label{eq:finite_block_triangular_identity}
D_{n,F}(z)
=
\frac1n\sum_{s=1}^{n}\sum_{t=0}^{s+L-1}\Delta_{t,s-1}(z).
\end{equation}
}

We next prove the locally uniform summability of the limiting terms. Fix $K\Subset\Omega$ and define
\begin{equation}
\label{eq:compact_domain_constants}
\begin{aligned}
\delta_K
&\triangleq
\inf_{\substack{z\in K,\,0\le r\le L:\,h_r>0\\0\le u_\ell\le A}}
\operatorname{Re}
\left(
\lambda_0+h_rz+\sum_{\ell\ne r}h_\ell u_\ell
\right)>0,\\
R_K
&\triangleq
\sup_{\substack{z\in K,\,0\le r\le L:\,h_r>0\\0\le u_\ell\le A}}
\left|
\lambda_0+h_rz+\sum_{\ell\ne r}h_\ell u_\ell
\right|<\infty .
\end{aligned}
\end{equation}
The positivity of $\delta_K$ follows from $K\Subset\Omega$. Therefore every affected complex Poisson kernel satisfies
\begin{equation}
\label{eq:complex_kernel_envelope_pre}
\left|e^{-\zeta}\frac{\zeta^y}{y!}\right|
\le
 e^{-\delta_K}\frac{R_K^y}{y!},
\qquad z\in K.
\end{equation}
This factorial envelope implies that all finite direct terms $t=0,\ldots,L$ are locally uniformly bounded on $K$.

For $t>L$, the symbol $X_0$ no longer appears in the physical intensity $Z_t$. Its influence on the information density can then pass only through the finite affected history consisting of $Y_0^L$ and the filter at time $L+1$. The envelope \eqref{eq:complex_kernel_envelope_pre} yields a finite complex signed measure $\Xi_{L+1}^{(z)}$ on that finite-history space and the bound
\begin{equation}
\label{eq:complex_initial_tv_bound}
\begin{aligned}
\sup_{z\in K}
\lVert\Xi_{L+1}^{(z)}\rVert_{\rm TV}
&\le
\widetilde C_K
\sum_{y_0^L\in\mathbb N_0^{L+1}}
\prod_{r=0}^{L}\frac{R_K^{y_r}}{y_r!} =\widetilde C_K e^{(L+1)R_K}<\infty .
\end{aligned}
\end{equation}
The measure $\Xi_{L+1}^{(z)}$ is the signed difference between the complexified finite-history law generated by the perturbed kernels and the baseline finite-history law. It is understood on the augmented history space of Lemma~\ref{lem:signed_initial_extension}, with physical state $Q_{L+1}$ and posterior baseline filtering-state coordinate $\widehat\Pi_{L+1}$ after the ordinary observation $Y_{L+1}$ has been processed. Starting from the affected history through $Y_0^L$, appending the ordinary state/observation update at time $L+1$ is a Markov push-forward and therefore does not increase total variation, so the bound \eqref{eq:complex_initial_tv_bound} remains valid. It has total mass zero, and the factorized bound in \eqref{eq:complex_initial_tv_bound} is independent of $F$.
{
The same construction applies to every finite-history functional. If the conditioning window starts after time $L$, then fixing $X_0$ does not affect any variable in that window and $\Delta_{t,m}(z)=0$. Otherwise, $\Delta_{L+1,m}(z)$ is uniformly bounded by \eqref{eq:complex_initial_tv_bound} together with the uniform logarithmic-moment bound. For $t>L+1$, the perturbation is propagated from the zero-mass augmented-state measure $\Xi_{L+1}^{(z)}$ through ordinary real kernels, so Lemma~\ref{lem:signed_initial_extension}, applied with $s=L+1$, gives a bound proportional to $\beta^{t-L-1}$. Enlarging $C_K$ to absorb both the single term $t=L+1$ and one factor of $\beta^{-1}$ yields
\begin{equation}
\label{eq:delta_tm_uniform_explicit}
\sup_{m\ge0}\sup_{z\in K}|\Delta_{t,m}(z)|
\le C_K\beta^{t-L},
\qquad t>L,
\end{equation}
where $\beta<1$ is independent of $K$ and $F$.
}

For fixed finite history length $m$, the function $\Delta_{t,m}(z)$ in \eqref{eq:Delta_tm_definition} is an absolutely convergent countable sum over $y_{t-m}^t$ and an integral over the finitely many input variables that can affect this window. Each summand is a product of affine-in-$z$ Poisson kernels and the fixed real logarithmic factor $\ell_{F,m}(y_t|y_{t-m}^{t-1})$. For a fixed finite history and a fixed observation sequence, this logarithmic factor is a constant independent of $z$; therefore, after multiplication by the complex finite measure generated by the analytically continued Poisson kernels, the finite-history expression is obtained as a locally uniformly dominated countable sum of holomorphic terms. The bound \eqref{eq:complex_kernel_envelope_pre} and the estimate $|\ell_{F,m}(y_t|y_{t-m}^{t-1})|\le C+C\sum_i y_i\log(y_i+1)$ give locally uniform domination; hence $\Delta_{t,m}$ is holomorphic on $\Omega$.

{
Fix $t$ and a compact $K\Subset\Omega$. For $m'>m>t+L+1$, compare the two baseline predictors associated with conditioning lengths $m'$ and $m$. They are driven by the same last $m$ observations, with possibly different initializations at the beginning of that window. Let $\pi_t^{(m')}$ and $\pi_t^{(m)}$ be their real predictive filters at time $t$, and set $D_{t,m,m'}\triangleq\lVert\pi_t^{(m')}-\pi_t^{(m)}\rVert_{\rm TV}$. By conditioning on the older observations and applying \eqref{eq:uff_tv_forgetting} to the common $m$-symbol suffix under the true baseline law,
\begin{equation}\label{eq:history_filter_comparison}
\begin{aligned}
\sup_{F\in\mathcal P([0,A])}\mathbb E_F[D_{t,m,m'}]
&\le C\beta^m,\\
\mathbb E_F[D_{t,m,m'}^2]
&\le2C\beta^m.
\end{aligned}
\end{equation}
where the second inequality uses $D_{t,m,m'}\le2$.

To control the complex perturbation explicitly, under the baseline joint input/output law put $Z_r^{(z)}\triangleq Z_r+h_r(z-X_0)$ for $0\le r\le L$ and define the finite-window likelihood ratio
\begin{equation}\label{eq:complex_window_likelihood_ratio}
W_z\triangleq \prod_{r=0}^{L}
\exp\!\bigl[-(Z_r^{(z)}-Z_r)\bigr]
\left(\frac{Z_r^{(z)}}{Z_r}\right)^{Y_r}.
\end{equation}
The baseline means satisfy $Z_r\ge\lambda_0$. The compact bound \eqref{eq:compact_domain_constants} therefore gives $\sup_{z\in K}|W_z|\le c_Kd_K^{\sum_{r=0}^{L}Y_r}$ for finite $c_K,d_K\ge1$ independent of $F$. Since all baseline Poisson means are at most $B_A$, their exponential moments imply, with $a\triangleq B_A/\lambda_0$,
\begin{equation}\label{eq:complex_weighted_moment}
\sup_{\substack{F\in\mathcal P([0,A])\\z\in K}}
\mathbb E_F\!\left[(|W_z|+1)^2a^{2Y_t}\right]\le C_K<\infty.
\end{equation}
For every output-dependent finite-history functional, fixing $X_0=z$ in the affected kernels is equivalent to integrating that functional against $W_z$ under the baseline law. With $\ell_{F,m,t}\triangleq-\log P_F(Y_t\mid Y_{t-m}^{t-1})$, we thus have the exact identity
\begin{equation}\label{eq:complex_finite_history_difference}
\Delta_{t,m'}(z)-\Delta_{t,m}(z)
=\mathbb E_F\!\left[(W_z-1)(\ell_{F,m',t}-\ell_{F,m,t})\right].
\end{equation}
The log-predictive Lipschitz bound \eqref{eq:uff_log_lip} yields $|\ell_{F,m',t}-\ell_{F,m,t}|\le e^{B_A-\lambda_0}a^{Y_t}D_{t,m,m'}$. Consequently, Cauchy--Schwarz and the preceding second-moment bounds give, uniformly over $F\in\mathcal P([0,A])$,
\begin{equation}\label{eq:complex_cauchy_schwarz_step}
\begin{aligned}
&\sup_{z\in K}|\Delta_{t,m'}(z)-\Delta_{t,m}(z)|\\
&\le e^{B_A-\lambda_0}
\sup_{z\in K}\left(\mathbb E_F[(|W_z|+1)^2a^{2Y_t}]\right)^{1/2}
\left(\mathbb E_F[D_{t,m,m'}^2]\right)^{1/2}\\
&\le C_K\beta^{m/2}.
\end{aligned}
\end{equation}
Writing $\widehat\beta\triangleq\sqrt{\beta}<1$, we have in particular
\begin{equation}
\label{eq:Delta_tm_cauchy}
\sup_{F\in\mathcal P([0,A])}\sup_{z\in K}
|\Delta_{t,m'}(z)-\Delta_{t,m}(z)|
\le C_K\widehat\beta^{m}.
\end{equation}
Thus $\{\Delta_{t,m}\}_{m\ge0}$ is locally uniformly Cauchy and we define
\begin{equation}
\label{eq:Delta_t_limit_definition}
\Delta_t(z)\triangleq\lim_{m\to\infty}\Delta_{t,m}(z).
\end{equation}
The convergence is locally uniform, so $\Delta_t$ is holomorphic by Weierstrass' theorem. Letting $m\to\infty$ in \eqref{eq:delta_tm_uniform_explicit} gives
\begin{equation}
\label{eq:delta_t_uniform_explicit}
\sup_{z\in K}|\Delta_t(z)|\le C_K\beta^{t-L},\qquad t>L,
\end{equation}
and hence, together with the finitely many direct terms, \eqref{eq:delta_series_uniform_bound_new}.

Set
\begin{equation}
A_s(z)\triangleq\sum_{t=0}^{s+L-1}\Delta_{t,s-1}(z).
\end{equation}
Then \eqref{eq:finite_block_triangular_identity} reads $D_{n,F}(z)=n^{-1}\sum_{s=1}^{n}A_s(z)$. We claim that
\begin{equation}
A_s(z)\longrightarrow\sum_{t=0}^{\infty}\Delta_t(z)
\end{equation}
locally uniformly on $K$. Indeed, given $\epsilon>0$, choose $T>L$ so that the two tails bounded by \eqref{eq:delta_tm_uniform_explicit} and \eqref{eq:delta_t_uniform_explicit} are each below $\epsilon$. For the finitely many indices $0\le t\le T$, the local uniform convergence $\Delta_{t,s-1}\to\Delta_t$ makes their total difference below $\epsilon$ for all sufficiently large $s$. Hence the displayed convergence follows. Ces\`aro averaging now gives
\begin{equation}
\label{eq:finite_block_derivative_limit}
D_{n,F}(z)
\to
\sum_{t=0}^{\infty}\Delta_t(z)
\end{equation}
locally uniformly on $K$. Because the constants in Lemmas~\ref{lem:uniform_filter_forgetting} and~\ref{lem:signed_initial_extension}, the factorial envelopes, and the weighted-moment bounds above are all independent of $F\in\mathcal P([0,A])$, the same argument is uniform in $F$ and proves \eqref{eq:finite_block_uniform_in_F}.
}
It remains to identify this analytic limit with the derivative of the entropy rate. Put $\bar H_n(F)\triangleq n^{-1}H_F(Y_1^n)$. By stationarity and the chain rule,
\begin{equation}
\bar H_n(F)=\frac1n\sum_{r=0}^{n-1}H_F(Y_0|Y_{-r}^{-1}),
\end{equation}
and \eqref{eq:uff_entropy_history}, now uniform over $F\in\mathcal P([0,A])$, gives (with the $r=0$ term absorbed into the constant)
\begin{equation}
\label{eq:block_entropy_uniform_rate}
\sup_{F\in\mathcal P([0,A])}|\bar H_n(F)-h_F(\mathbf Y)|
\le \frac{C}{n}\sum_{r=0}^{n-1}\beta^r
\le \frac{C}{n(1-\beta)}.
\end{equation}
Fix real $x\in[0,A]$ and set $F_s=(1-s)F+s\delta_x$. For finite $n$, ordinary finite-dimensional differentiation gives
\begin{equation}
\label{eq:path_integral_pointmass_finite}
\bar H_n(F_\epsilon)-\bar H_n(F)
=\int_0^\epsilon \frac{D_{n,F_s}(x)}{1-s}\,ds,
\end{equation}
because $dF_s/ds=(\delta_x-F_s)/(1-s)$. Taking $n\to\infty$ in \eqref{eq:path_integral_pointmass_finite} using \eqref{eq:block_entropy_uniform_rate} and the uniform convergence \eqref{eq:finite_block_uniform_in_F} yields
\begin{equation}
\label{eq:path_integral_pointmass_rate}
h_{F_\epsilon}(\mathbf Y)-h_F(\mathbf Y)
=\int_0^\epsilon \frac{\dot h_{Y,F_s}^{\rm hol}(x)}{1-s}\,ds.
\end{equation}
For each finite $n$, $F\mapsto D_{n,F}(x)$ is weakly continuous by finite-dimensional dominated convergence; hence \eqref{eq:finite_block_uniform_in_F} implies that $F\mapsto\dot h_{Y,F}^{\rm hol}(x)$ is continuous. Dividing \eqref{eq:path_integral_pointmass_rate} by $\epsilon$ and letting $\epsilon\downarrow0$ proves
\begin{equation}
\left.\frac{d}{d\epsilon}h_{(1-\epsilon)F+\epsilon\delta_x}(\mathbf Y)\right|_{\epsilon=0^+}
=\dot h_{Y,F}^{\rm hol}(x).
\end{equation}
Thus the limit is the output-entropy-rate directional derivative at real point masses in the peak-constrained alphabet.
\end{proof}

\begin{lemma}
\label{lem:holomorphic_influence}
For $F,G\in\mathcal P_A(P_{\rm avg})$, define $F_\epsilon\triangleq(1-\epsilon)F+\epsilon G$ and
\begin{equation}
D\mathcal I_F(G-F)
\triangleq
\left.
\frac{d}{d\epsilon}\mathcal I(F_\epsilon)
\right|_{\epsilon=0^+},
\end{equation}
whenever the derivative exists. For $x\in[0,A]$, define the ambient point-mass influence
\begin{equation}
\label{eq:psi_real_definition}
\psi_F(x)
\triangleq
\left.
\frac{d}{d\epsilon}\mathcal I((1-\epsilon)F+\epsilon\delta_x)
\right|_{\epsilon=0^+}.
\end{equation}
Then $\psi_F$ exists and is continuous on $[0,A]$. Moreover, there is a holomorphic function $\Psi_F$ on $\Omega$ such that $\Psi_F(x)=\psi_F(x)$ for $x\in[0,A]$, and for every $G\in\mathcal P_A(P_{\rm avg})$, we have
\begin{equation}
\label{eq:first_variation_integral_form}
D\mathcal I_F(G-F)
=
\int_{[0,A]}\psi_F(x)dG(x)
-
\int_{[0,A]}\psi_F(x)dF(x).
\end{equation}
\end{lemma}

\begin{proof}
Lemma~\ref{lem:finite_block_derivative} gives a locally uniformly convergent holomorphic series for the output-entropy part. It remains to include the conditional entropy-rate part. Since $h_F(\mathbf Y|\mathbf X)=\mathbb E_F[\phi(Z_0)]$, replacing one input symbol by $z$ affects the $L+1$ output intensities containing that symbol. Define
\begin{equation}
\label{eq:conditional_influence_formula}
\begin{aligned}
\dot h_{Y|X,F}^{\rm hol}(z)
&\triangleq
\sum_{r=0}^{L}
\mathbb E_F
\left[
\phi\left(
\lambda_0+h_rz+\sum_{\ell\ne r}h_\ell X_{-\ell}
\right)
-
\phi\left(
\lambda_0+\sum_{\ell=0}^{L}h_\ell X_{-\ell}
\right)
\right].
\end{aligned}
\end{equation}
For $z\in\Omega$, each argument containing a positive tap has positive real part, and zero-tap terms are independent of $z$. The Poisson entropy function admits the analytic representation
\begin{equation}
\phi(\zeta)=\zeta-\zeta\log\zeta+
\sum_{m=0}^{\infty}e^{-\zeta}\frac{\zeta^m}{m!}\log(m!),
\end{equation}
with the principal logarithm. For every compact $K\Subset\Omega$, the arguments of $\phi$ in \eqref{eq:conditional_influence_formula} lie in a compact subset of the open right half-plane. Hence they do not meet the nonpositive real axis, and the principal logarithm is single-valued and holomorphic on a neighborhood of that compact set. On such compact subsets, the series is dominated by a factorial tail times $m\log(m+1)$. Hence $\dot h_{Y|X,F}^{\rm hol}$ is holomorphic on $\Omega$.

We define
\begin{equation}
\label{eq:Psi_def}
\Psi_F(z)
\triangleq
\dot h_{Y,F}^{\rm hol}(z)
-
\dot h_{Y|X,F}^{\rm hol}(z),
\qquad z\in\Omega .
\end{equation}
By Lemma~\ref{lem:finite_block_derivative} and the preceding paragraph, $\Psi_F$ is holomorphic. For real $x\in[0,A]$, both terms are genuine directional derivatives, so $\Psi_F(x)=\psi_F(x)$. Continuity of $\psi_F$ on $[0,A]$ follows from the continuity of the holomorphic function $\Psi_F$.

It remains to justify the general direction without formally interchanging a derivative and an entropy-rate limit. Set $F_s=(1-s)F+sG$. For the output entropy, finite-dimensional linearity of the first variation gives
\begin{equation}
\label{eq:path_integral_general_finite}
\bar H_n(F_\epsilon)-\bar H_n(F)
=\int_0^\epsilon\frac{1}{1-s}
\int_{[0,A]}D_{n,F_s}(x)\,dG(x)\,ds.
\end{equation}
Using \eqref{eq:block_entropy_uniform_rate} and the convergence in \eqref{eq:finite_block_uniform_in_F}, which is uniform for $x\in[0,A]$ and $F_s\in\mathcal P([0,A])$, we may let $n\to\infty$ to obtain
\begin{equation}
\label{eq:path_integral_general_output_rate}
h_{F_\epsilon}(\mathbf Y)-h_F(\mathbf Y)
=\int_0^\epsilon\frac{1}{1-s}
\int_{[0,A]}\dot h_{Y,F_s}^{\rm hol}(x)\,dG(x)\,ds.
\end{equation}
The conditional-entropy term is finite-dimensional in the input window, and the same product-measure differentiation applied directly to \eqref{eq:conditional_influence_formula} gives
\begin{equation}
\label{eq:path_integral_general_cond_rate}
h_{F_\epsilon}(\mathbf Y|\mathbf X)-h_F(\mathbf Y|\mathbf X)
=\int_0^\epsilon\frac{1}{1-s}
\int_{[0,A]}\dot h_{Y|X,F_s}^{\rm hol}(x)\,dG(x)\,ds.
\end{equation}
The integrands are continuous in $s$; for the output part this follows from the uniform limit of the finite-dimensional continuous functions in \eqref{eq:finite_block_uniform_in_F}, and for the conditional part from \eqref{eq:conditional_influence_formula}. Subtracting \eqref{eq:path_integral_general_cond_rate} from \eqref{eq:path_integral_general_output_rate}, dividing by $\epsilon$, and letting $\epsilon\downarrow0$ yields
\begin{equation}
D\mathcal I_F(G-F)=\int_{[0,A]}\Psi_F(x)\,dG(x).
\end{equation}
Taking $G=F$ gives $\int\Psi_F(x)dF(x)=0$. Since $\Psi_F(x)=\psi_F(x)$ on $[0,A]$, this is exactly \eqref{eq:first_variation_integral_form}.
\end{proof}

\subsection{KKT Condition for the i.i.d. Marginal Law}

We now derive the KKT condition from the first-variation formula. The boundary cases of the average constraint are included for completeness.

\begin{lemma}
\label{lem:kkt}
Assume $P_{\rm avg}>0$, and let $F^\star$ be a maximizer of $\mathcal I(F)$ over $\mathcal P_A(P_{\rm avg})$. Then there exist constants $\lambda\ge0$ and $\gamma\in\mathbb R$ such that
\begin{equation}
\label{eq:kkt_inequality_app}
\psi_{F^\star}(x)-\lambda x\le\gamma,
\qquad x\in[0,A],
\end{equation}
and
\begin{equation}
\label{eq:kkt_equality_support_app}
\psi_{F^\star}(x)-\lambda x=\gamma,
\qquad x\in\operatorname{supp}(F^\star).
\end{equation}
Moreover,
\begin{equation}
\label{eq:kkt_complementary_slackness_app}
\lambda\left(\int x\,dF^\star(x)-P_{\rm avg}\right)=0.
\end{equation}
\end{lemma}

\begin{proof}
If $P_{\rm avg}\ge A$, the average constraint is redundant because every probability law on $[0,A]$ satisfies $\int x\,dF(x)\le A\le P_{\rm avg}$. In this case the feasible set is simply $\mathcal P([0,A])$. Since $F^\star$ is optimal, Lemma~\ref{lem:holomorphic_influence} gives
\begin{equation}
\int\psi_{F^\star}(x)dG(x)
\le
\int\psi_{F^\star}(x)dF^\star(x),
\qquad G\in\mathcal P([0,A]).
\end{equation}
Thus \eqref{eq:kkt_inequality_app}-\eqref{eq:kkt_complementary_slackness_app} hold with $\lambda=0$ and $\gamma=\int\psi_{F^\star}dF^\star$.

It remains to consider $0<P_{\rm avg}<A$. The feasible set has a strict interior point for the average constraint: the law $\delta_0$ satisfies $\int x\,d\delta_0(x)=0<P_{\rm avg}$. This is the Slater condition for the single affine inequality constraint. By Lemma~\ref{lem:holomorphic_influence}, $\psi_{F^\star}$ is continuous and bounded on $[0,A]$, so $G\mapsto\int\psi_{F^\star}dG$ is weakly continuous. The variational inequality at the maximizer is
\begin{equation}
\label{eq:kkt_variational_pre_slater}
\int\psi_{F^\star}(x)dG(x)
\le
\int\psi_{F^\star}(x)dF^\star(x),
\qquad G\in\mathcal P_A(P_{\rm avg}).
\end{equation}
Applying the normal-cone/KKT theorem to this resulting linear optimization problem over probability measures, with its continuous affine inequality constraint and Slater's condition~\cite{rockafellar1970convex}, gives a multiplier $\lambda\ge0$ such that
\begin{equation}
\label{eq:normal_cone_condition}
\int(\psi_{F^\star}(x)-\lambda x)dG(x)
\le
\int(\psi_{F^\star}(x)-\lambda x)dF^\star(x)
\end{equation}
for every probability measure $G$ on $[0,A]$, together with complementary slackness \eqref{eq:kkt_complementary_slackness_app}. Define
\begin{equation}
\gamma\triangleq
\int(\psi_{F^\star}(u)-\lambda u)dF^\star(u).
\end{equation}
Taking $G=\delta_x$ in \eqref{eq:normal_cone_condition} gives \eqref{eq:kkt_inequality_app}. Integrating \eqref{eq:kkt_inequality_app} with respect to $F^\star$ gives equality, hence $\psi_{F^\star}(x)-\lambda x=\gamma$ for $F^\star$-almost every $x$. Since $\psi_{F^\star}$ is continuous, the equality extends to the closed support of $F^\star$. This proves \eqref{eq:kkt_equality_support_app}.
\end{proof}

\subsection{Supralinear Growth of the Holomorphic Extension}

It remains to prove the growth estimate that prevents the KKT defect from vanishing identically. The estimate is stated for $\Psi_F$, the holomorphic extension introduced in Lemma~\ref{lem:holomorphic_influence}. For $x>A$, $\Psi_F(x)$ is an analytic-continuation value rather than a feasible point-mass derivative.

\begin{lemma}
\label{lem:supralinear_growth}
For every $F\in\mathcal P_A(P_{\rm avg})$, we have
\begin{equation}
\label{eq:Psi_supralinear_statement}
\Psi_F(x)\ge x\log x-O(x),
\qquad x\to+\infty.
\end{equation}
Consequently, for every finite $\lambda\ge0$,
$\Psi_F(x)-\lambda x\to+\infty$ as $x\to+\infty$.
\end{lemma}

\begin{proof}
The proof evaluates the analytic continuation on the positive real axis. Although the point $x>A$ is outside the feasible input alphabet, the finite-dimensional Poisson kernels defining $\Psi_F(x)$ are ordinary real Poisson kernels and are therefore probabilistically interpretable as a comparison experiment in which one symbol is set to the larger value $x$.

Under the baseline input law $F$, every channel intensity satisfies $\lambda_0\le Z_t\le B_A$. Hence, for every output history $Y_{<t}$, we have
\begin{equation}
\label{eq:baseline_predictive_upper}
P_F(Y_t=y|Y_{<t})
\le
 e^{-\lambda_0}\frac{B_A^y}{y!} .
\end{equation}
Therefore $-\log P_F(Y_t=y|Y_{<t})\ge \lambda_0-y\log B_A+\log(y!)$.

Set the distinguished symbol $X_0$ equal to $x$ and keep all other symbols i.i.d. according to $F$. For $r=0,1,\ldots,L$ with $h_r>0$, the output $Y_r$ has conditional Poisson mean
\begin{equation}
\mu_r(x)=\lambda_0+h_rx+\sum_{\ell\ne r}h_\ell X_{r-\ell}
=h_rx+O(1),
\qquad x\to+\infty,
\end{equation}
uniformly over the remaining bounded input symbols. Using \eqref{eq:baseline_predictive_upper} under this comparison experiment gives
\begin{equation}
\begin{aligned}
&\mathbb E_F^{(0,x)}[-\log P_F(Y_r|Y_{<r})] \ge
\lambda_0-(\log B_A)\mathbb E_F^{(0,x)}[Y_r]
+\mathbb E_F^{(0,x)}[\log(Y_r!)] .
\end{aligned}
\end{equation}
Stirling's formula and standard Poisson moment estimates (applied conditionally on the remaining input symbols and then averaged) imply
\begin{equation}
\label{eq:poisson_log_factorial_moment}
\begin{aligned}
\mathbb E_F^{(0,x)}[\log(Y_r!)]
&=\mathbb E_F^{(0,x)}[\mu_r(x)\log\mu_r(x)]+O(x)\\
&=h_rx\log x+O(x).
\end{aligned}
\end{equation}
Indeed, the inequality $\log(y!)\ge y\log y-y$ for $y\ge1$, together with the Chernoff bound $\mathsf P\{Y_r<\mu_r(x)/2\}\le e^{-c\mu_r(x)}$ and $\mathbb E[Y_r]=\mu_r(x)$, gives the lower estimate $\mathbb E[\log(Y_r!)]\ge \mu_r(x)\log\mu_r(x)-O(\mu_r(x))$. The matching upper estimate follows from the standard Stirling upper bound and $\mathbb E[Y_r\log(Y_r+1)]=\mu_r(x)\log\mu_r(x)+O(\mu_r(x))$. The constants are uniform over the remaining input symbols because $\mu_r(x)=h_rx+O(1)$ uniformly in those bounded symbols. Here the intermediate Poisson probability and expectation bounds are conditional on the other input symbols; the displayed result follows by taking their outer expectation. When $h_r=0$, the perturbed mean of $Y_r$ remains in $[\lambda_0,B_A]$, so the corresponding contribution is bounded below by $-O(1)$ and the same inequality below remains valid.
The corresponding baseline expectation is uniformly $O(1)$ because all baseline means lie in $[\lambda_0,B_A]$ and is absorbed into the $O(x)$ term. Since $\mathbb E_F^{(0,x)}[Y_r]=h_rx+O(1)$, the direct contribution satisfies
\begin{equation}
\Delta_r(x)
\ge
h_rx\log x-O(x),
\qquad r=0,\ldots,L.
\end{equation}
Summing over the directly affected outputs and using $\sum_{r=0}^L h_r=1$ gives
\begin{equation}
\label{eq:direct_terms_sum_growth}
\sum_{r=0}^{L}\Delta_r(x)
\ge
x\log x-O(x).
\end{equation}

For $t>L$, the large symbol $X_0=x$ no longer appears in the physical intensity $Z_t$. Its future influence is only through the finite history generated by the affected observations $Y_0^L$ and the baseline filter. For positive real $x$, propagate the perturbed and baseline affected-history laws through the ordinary state/observation update at time $L+1$. The resulting augmented-state laws on $(Q_{L+1},\widehat\Pi_{L+1})$ are probability laws, so their signed difference has total mass zero and total variation at most $2$. The single term $\Delta_{L+1}(x)$ is uniformly $O(1)$ by the uniform logarithmic-moment bound. For $t>L+1$, Lemma~\ref{lem:signed_initial_extension}, applied from the posterior augmented state at $s=L+1$, gives $|\Delta_t(x)|\le C\beta^{t-L-1}$. Enlarging $C$ absorbs $t=L+1$ and one factor of $\beta^{-1}$, so
\begin{equation}
\label{eq:tail_growth_bound}
|\Delta_t(x)|\le C\beta^{t-L},
\qquad t>L,
\end{equation}
where $C<\infty$ and $0<\beta<1$ do not depend on $x$ or $F$. Hence $\sum_{t>L}\Delta_t(x)=O(1)$, and
\begin{equation}
\label{eq:output_entropy_influence_growth}
\dot h_{Y,F}^{\rm hol}(x)
=\sum_{t=0}^{\infty}\Delta_t(x)
\ge x\log x-O(x).
\end{equation}

The conditional entropy-rate part grows only logarithmically. Terms with $h_r=0$ are independent of $x$. For $h_r>0$, the corresponding argument in \eqref{eq:conditional_influence_formula} equals $h_rx+O(1)$ uniformly over the bounded remaining input symbols, and the entropy of a Poisson random variable satisfies
$\phi(\mu)=\frac12\log(2\pi e\mu)+O(\mu^{-1})$ as $\mu\to+\infty$. Each positive-tap term is therefore $O(\log x)$, and there are only $L+1$ terms. Hence
\begin{equation}
\label{eq:conditional_influence_log_growth}
\dot h_{Y|X,F}^{\rm hol}(x)=O(\log x).
\end{equation}
Combining \eqref{eq:output_entropy_influence_growth} and \eqref{eq:conditional_influence_log_growth} with $\Psi_F=\dot h_{Y,F}^{\rm hol}-\dot h_{Y|X,F}^{\rm hol}$ proves \eqref{eq:Psi_supralinear_statement}. The final claim follows immediately after subtracting any finite linear term $\lambda x$.
\end{proof}

\subsection{Completion of the Proof of Theorem~\ref{thm:finite_support}}

If $P_{\rm avg}=0$, the only feasible input law is $\delta_0$, so the theorem is immediate. Assume henceforth that $P_{\rm avg}>0$. Let $F^\star$ be a rate-maximizing i.i.d. distribution, whose existence is guaranteed by Lemma~\ref{lem:existence}. By Lemma~\ref{lem:kkt}, there exist $\lambda\ge0$ and $\gamma\in\mathbb R$ such that $\psi_{F^\star}(x)-\lambda x\le\gamma$ on $[0,A]$ and equality holds on $\operatorname{supp}(F^\star)$.

Define the holomorphic KKT defect
\begin{equation}
J(z)\triangleq\Psi_{F^\star}(z)-\lambda z-\gamma,
\qquad z\in\Omega .
\end{equation}
On the feasible interval, $\Psi_{F^\star}(x)=\psi_{F^\star}(x)$; hence $J(x)=0$ for every $x\in\operatorname{supp}(F^\star)$. If $\operatorname{supp}(F^\star)$ were infinite, compactness of $[0,A]$ would give an accumulation point $x_0\in[0,A]\subset\Omega$. The identity theorem for holomorphic functions would then imply $J(z)\equiv0$ on $\Omega$. In particular,
\begin{equation}
\label{eq:J_zero_positive_axis}
J(x)=0,
\qquad x\ge0.
\end{equation}
However, Lemma~\ref{lem:supralinear_growth} gives
\begin{equation}
J(x)=\Psi_{F^\star}(x)-\lambda x-\gamma
\ge x\log x-O(x)-\lambda x-\gamma\to+\infty,
\end{equation}
which contradicts \eqref{eq:J_zero_positive_axis}. Hence $\operatorname{supp}(F^\star)$ cannot be infinite. Every rate-maximizing i.i.d. input distribution therefore has finite support on $[0,A]$.

{
\subsection{Proof of Proposition~\ref{prop:small_peak_isi_binary}}

\begin{proof}
Fix an auxiliary peak level $\bar A>0$ and consider all input
laws on $[0,\bar A]$.  The influence function $\Psi_F$ from
Lemma~\ref{lem:holomorphic_influence} is holomorphic on the common
domain $\Omega$ of Lemma~\ref{lem:finite_block_derivative}.
At $F=\delta_0$ the baseline outputs are i.i.d. with law
$\operatorname{Poisson}(\lambda_0)$.  Fixing one input symbol to
$z$ changes only the $L+1$ Poisson means
$\lambda_0+h_\ell z$, $0\le\ell\le L$.  Subtracting the
conditional-entropy influence from the output-entropy influence,
or equivalently evaluating the corresponding Poisson relative
entropies, gives
\begin{equation}
\label{eq:small_peak_delta0_influence}
\begin{aligned}
\Psi_{\delta_0}(z)
&=\sum_{\ell=0}^{L}
D\!\left(\operatorname{Poisson}(\lambda_0+h_\ell z)
\,\middle\|\,\operatorname{Poisson}(\lambda_0)\right)\\
&=\sum_{\ell=0}^{L}\left[
(\lambda_0+h_\ell z)
\log\!\left(1+\frac{h_\ell z}{\lambda_0}\right)
-h_\ell z\right],
\end{aligned}
\end{equation}
where the first line is interpreted through its analytic
continuation when $z$ is nonreal.  Hence, for real $x\ge0$,
\begin{equation}
\label{eq:small_peak_delta0_curvature}
\Psi_{\delta_0}''(x)
=\sum_{\ell=0}^{L}
\frac{h_\ell^2}{\lambda_0+h_\ell x}>0.
\end{equation}
In particular, the positive delayed taps contribute explicitly
to the curvature.

We next show that this positive curvature persists uniformly
for all sufficiently small-peak input laws.  The uniform
convergence in \eqref{eq:finite_block_uniform_in_F}, applied with
the \emph{fixed} ambient peak $\bar A$, and the finite-dimensional
weak continuity of $F\mapsto D_{n,F}(z)$ imply
\begin{equation}
\label{eq:small_peak_weak_hol_continuity}
F_j\Rightarrow\delta_0
\quad\Longrightarrow\quad
\sup_{z\in K}|\Psi_{F_j}(z)-\Psi_{\delta_0}(z)|\to0
\end{equation}
for each compact $K\Subset\Omega$.  Indeed, for a fixed
$n$ the finite-block expression is jointly continuous in the
weak input law and $z\in K$, by the uniform Poisson tails and
the baseline log-probability bounds used in Lemma~\ref{lem:finite_block_derivative};
then \eqref{eq:finite_block_uniform_in_F} gives the output-entropy
part of \eqref{eq:small_peak_weak_hol_continuity}.  The
conditional-entropy part follows directly from
\eqref{eq:conditional_influence_formula}.  Importantly, the
constants are taken on $[0,\bar A]$, not on a varying input
alphabet.

Choose $r>0$ so that $2r<\lambda_0/h_+$ and $r<\bar A$,
where $h_+$ is as defined in
Lemma~\ref{lem:finite_block_derivative}.  Then the closed disk
$K=\{z:|z|\le2r\}$ lies in $\Omega$, and, by Cauchy's
estimate, for every $F\in\mathcal P([0,\bar A])$,
\begin{equation}
\label{eq:small_peak_cauchy_curvature}
\sup_{0\le x\le r}
|\Psi_F''(x)-\Psi_{\delta_0}''(x)|
\le\frac{2}{r^2}
\sup_{|z|\le2r}|\Psi_F(z)-\Psi_{\delta_0}(z)|.
\end{equation}
As $A\downarrow0$, every law supported on $[0,A]$ converges
to $\delta_0$ uniformly in the weak topology: for each
continuous $f$ on $[0,\bar A]$,
\begin{equation}
\label{eq:small_peak_weak_collapse}
\sup_{F\in\mathcal P([0,A])}
\left|\int f\,dF-f(0)\right|
\le\sup_{0\le x\le A}|f(x)-f(0)|\to0.
\end{equation}
Combining \eqref{eq:small_peak_weak_hol_continuity}--
\eqref{eq:small_peak_weak_collapse} with
\eqref{eq:small_peak_delta0_curvature}, there is
$A_{\rm sp}\in(0,r)$ such that, for all $A<A_{\rm sp}$,
all $F\in\mathcal P([0,A])$, and all $x\in[0,A]$,
\begin{equation}
\label{eq:small_peak_uniform_convexity}
\Psi_F''(x)\ge
\frac12\sum_{\ell=0}^{L}
\frac{h_\ell^2}{\lambda_0+h_\ell r}>0.
\end{equation}
Thus $\psi_F=\Psi_F|_{[0,A]}$ is strictly convex on
$[0,A]$ for every such law $F$.

Let $F^\star$ be an optimizer, whose existence follows from
Lemma~\ref{lem:existence}. Since $P_{\rm avg}>0$,
Lemma~\ref{lem:kkt} supplies $\eta\ge0$ and
$\gamma\in\mathbb R$ such that
$\psi_{F^\star}(x)-\eta x\le\gamma$ throughout $[0,A]$,
with equality at every support point.  By
\eqref{eq:small_peak_uniform_convexity}, the function
$x\mapsto\psi_{F^\star}(x)-\eta x$ is strictly convex; its
maximum over $[0,A]$ can therefore be attained only at
$0$ or $A$. Consequently,
$\operatorname{supp}(F^\star)\subseteq\{0,A\}$.

Finally, a feasible nondegenerate on--off input law exists
because $P_{\rm avg}>0$: choose
$0<p<\min\{1,P_{\rm avg}/A\}$ and use $F_p$ in
\eqref{eq:small_peak_binary_family}.  Under this law,
$I_{F_p}(X_0;Y_0)>0$, because the conditional output laws
for $X_0=0$ and $X_0=A$ have different means, separated by
$h_0A>0$.  By the chain rule and the independence of $X_k$
from $Y_1^{k-1}$,
\begin{equation}
\label{eq:small_peak_positive_rate}
\begin{aligned}
\mathcal I(F_p)
&=\lim_{n\to\infty}\frac1n I_{F_p}(X_1^n;Y_1^n)\\
&\ge \lim_{n\to\infty}\frac1n
\sum_{k=1}^{n} I_{F_p}(X_k;Y_k)
=I_{F_p}(X_0;Y_0)>0.
\end{aligned}
\end{equation}
A deterministic input has zero mutual information rate, so
$F^\star$ cannot be a one-point law.  Both endpoints must
therefore have positive mass, which proves
\eqref{eq:small_peak_binary_support} and the reduction
\eqref{eq:small_peak_binary_family}.
\end{proof}
}

\section{Numerical Results and Discussion}\label{sec:numerical_results}

\begin{table}[t]
\caption{Finite-grid parameters used in the numerical studies.}
\label{tab:numerical_parameters}
\centering
\footnotesize
\begin{tabular}{@{}p{0.34\columnwidth}@{\hspace{0.02\columnwidth}}p{0.58\columnwidth}@{}}
\toprule
Quantity & Value used in the plots \\
\midrule
Memoryless preliminary grid & $1501$ input points \\
Memoryless refinement & Continuous SQP over at most $6$ active mass points, $10$ starts \\
ISI grid size & $31$ input points \\
ISI forward block length & $N_b=3$ \\
Main ISI taps & $h=(0.65,0.35)$ \\
SQP starts for ISI plots & $16$ starts; $8$ starts for each curve in Fig.~\ref{fig:parameter_sensitivity_refined}(a) \\
Stopping tolerance & $10^{-11}$ for memoryless BA; $10^{-9}$ for ISI SQP objective \\
Active-mass display threshold & $10^{-4}$ \\
\bottomrule
\end{tabular}
\end{table}

In this section, we present finite-grid numerical studies for the i.i.d. finite-memory Poisson channel. These simulations are used to illustrate performance trends after the continuous input interval is discretized. They should not be interpreted as a numerical proof of Theorem~\ref{thm:finite_support}, since Theorem~\ref{thm:finite_support} is a continuous-alphabet structural result. In each simulation, the interval $[0,A]$ is quantized into a uniform grid and the resulting finite-alphabet problem is optimized under the average constraint. For the memoryless reference channel, we first use a dense-grid cost-constrained Blahut-Arimoto (BA) iteration to locate active clusters, and then refine the active mass locations and probabilities by continuous sequential quadratic programming (SQP). This refinement is used only in the memoryless support-evolution plot to avoid displaying one interior mass point as several neighboring grid points. For the two-tap ISI examples, the output entropy rate is approximated by a finite-block forward recursion for the hidden Markov output process. For a fixed i.i.d. input law, the recursion evaluates
\begin{equation}
\label{eq:numerical_block_rate}
\begin{aligned}
 \widehat{\mathcal I}_{N_b}(F)
 &\triangleq
 H_F(Y_{N_b}|Y_1^{N_b-1})
 -H_F(Y_{N_b}|X_{N_b},X_{N_b-1}),
\end{aligned}
\end{equation}
where $N_b$ is the finite block length used in the forward recursion. The discretized probabilities are optimized under the average constraint by sequential quadratic programming with multiple feasible initial points, including endpoint distributions and random feasible distributions. The best final value is retained. Photon counts are truncated to $\{0,1,\ldots,N_{\max}\}$, where
\begin{equation}
\label{eq:numerical_truncation_rule}
\begin{aligned}
N_{\max}
=
\max\left\{
35,
\left\lceil
\mu_{\max}+10\sqrt{\max\{\mu_{\max},1\}}
\right\rceil
\right\},
\end{aligned}
\end{equation}
$\mu_{\max}$ is the largest possible Poisson mean in the corresponding finite-grid computation, and each truncated Poisson column is renormalized. Unless otherwise specified, the finite-grid computations use the truncation and optimization settings summarized in Table~\ref{tab:numerical_parameters}.

Fig.~\ref{fig:support_evolution_refined} shows the active normalized amplitudes $x/A$ obtained from the finite-grid optimizers when the peak constraint varies over $A=0,1,\ldots,18$. The memoryless reference case is shown in Fig.~\ref{fig:support_evolution_refined}(a), while the two-tap ISI case with $h=(0.65,0.35)$ is shown in Fig.~\ref{fig:support_evolution_refined}(b). In both cases, only a small number of grid amplitudes remain active after optimization. The two endpoints $x/A=0$ and $x/A=1$ are selected throughout the displayed range, whereas interior amplitudes appear only for some peak-intensity values. Compared with the memoryless reference, the considered ISI case exhibits fewer visible interior branches.

\begin{figure}
 \centering
 \subfigure[Memoryless reference with $\lambda_0=1$ and $P_{\rm avg}=0.5A$.]{%
 \includegraphics[width=0.6\textwidth]{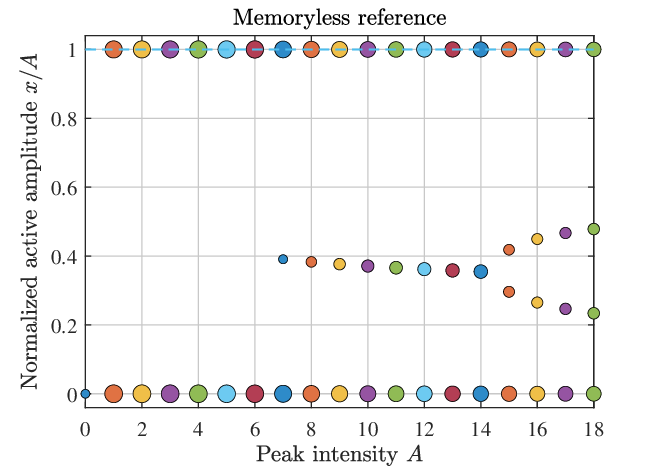}%
 \label{fig:support_evolution_refined_a}}
 \hfil
 \subfigure[Two-tap ISI with $h=(0.65,0.35)$, $\lambda_0=1$, and $P_{\rm avg}=0.5A$.]{%
 \includegraphics[width=0.6\textwidth]{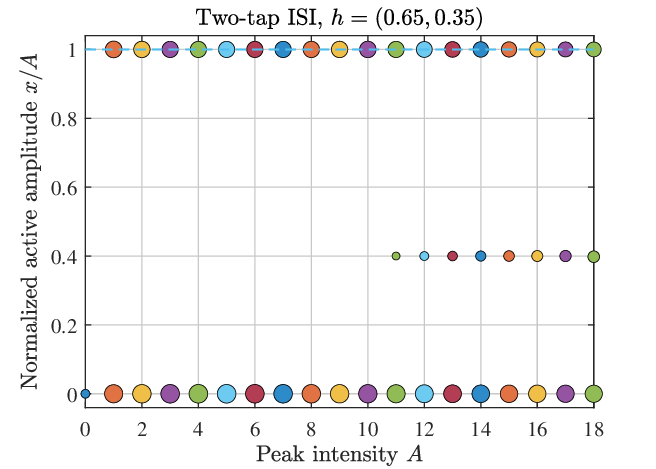}%
 \label{fig:support_evolution_refined_b}}
 \caption{Finite-grid active-amplitude evolution for $A=0,1,\ldots,18$. The vertical axis is normalized as $x/A$, and marker size indicates optimized probability mass.}
 \label{fig:support_evolution_refined}
\end{figure}

Fig.~\ref{fig:parameter_sensitivity_refined} compares different physical parameters in the two-tap ISI setting. In Fig.~\ref{fig:parameter_sensitivity_refined}(a), the finite-grid rate estimate increases with the peak intensity and decreases when the dark current $\lambda_0$ becomes larger. This is consistent with the reduced signal-to-background contrast under stronger background noise. In Fig.~\ref{fig:parameter_sensitivity_refined}(b), the peak intensity, dark current, and average-to-peak ratio are fixed, while the two-tap profile is changed. The rate estimate decreases as the channel memory becomes more spread. The annotations report the number $K$ of active grid amplitudes and the tap energy $H_2\triangleq\sum_{\ell}h_\ell^2$.

\begin{figure}
 \centering
 \subfigure[Rate estimate versus peak intensity for $h=(0.65,0.35)$, $P_{\rm avg}=0.5A$, and several dark-current levels.]{%
 \includegraphics[width=0.6\textwidth]{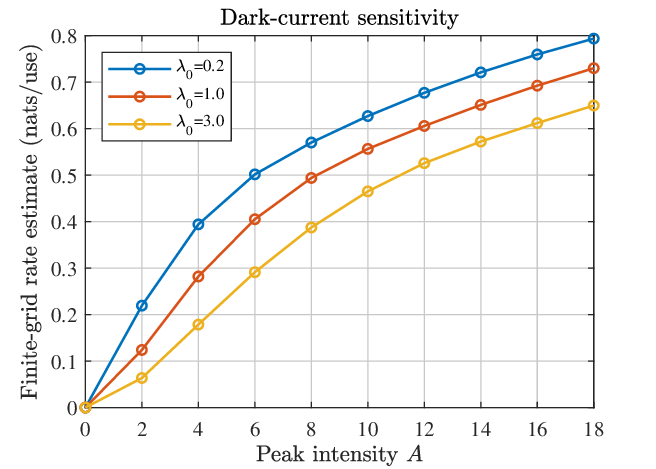}%
 \label{fig:parameter_sensitivity_refined_a}}
 \hfil
 \subfigure[Rate estimate versus tap profile at $A=18$, $\lambda_0=1$, and $P_{\rm avg}=0.5A$.]{%
 \includegraphics[width=0.6\textwidth]{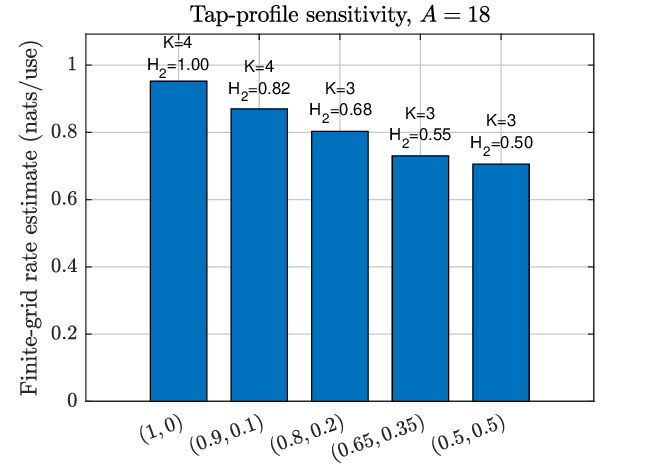}%
 \label{fig:parameter_sensitivity_refined_b}}
 \caption{Parameter sensitivity of the discretized i.i.d. finite-memory Poisson optimization.}
 \label{fig:parameter_sensitivity_refined}
\end{figure}

Fig.~\ref{fig:endpoint_behavior_refined} further studies the endpoint behavior for $A=15$ and $h=(0.65,0.35)$. The endpoint mass in Fig.~\ref{fig:endpoint_behavior_refined}(a) is the optimized probability assigned to the two endpoint grid points, namely $p(0)+p(A)$. The remaining probability is regarded as the interior mass. For the tested average-to-peak ratios, the optimized finite-grid laws allocate most of their mass to the two endpoints, and the endpoint mass becomes larger when the dark current increases. Fig.~\ref{fig:endpoint_behavior_refined}(b) compares the finite-grid optimum with two endpoint-only benchmarks. The optimized endpoint law maximizes the finite-block estimate over all endpoint distributions satisfying $0\le p(A)\le P_{\rm avg}/A$, whereas the full-average endpoint law fixes $p(A)=P_{\rm avg}/A$. The optimized endpoint benchmark is close to the finite-grid optimum over a large part of the tested range, especially for larger $\lambda_0$. The full-average endpoint law, however, can lose performance when the average constraint is loose. This indicates that endpoint support and full average-power usage are different numerical phenomena.

\begin{figure}
 \centering
 \subfigure[Endpoint and interior mass of the finite-grid optimizer.]{%
 \includegraphics[width=0.6\textwidth]{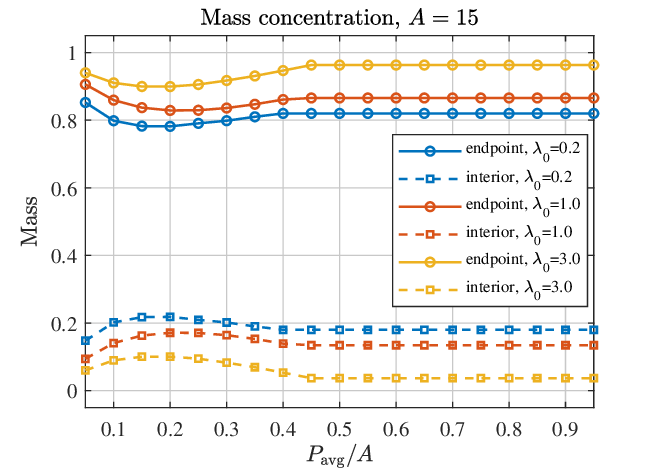}%
 \label{fig:endpoint_behavior_refined_a}}
 \hfil
 \subfigure[Endpoint-signaling benchmark relative to the finite-grid optimizer.]{%
 \includegraphics[width=0.6\textwidth]{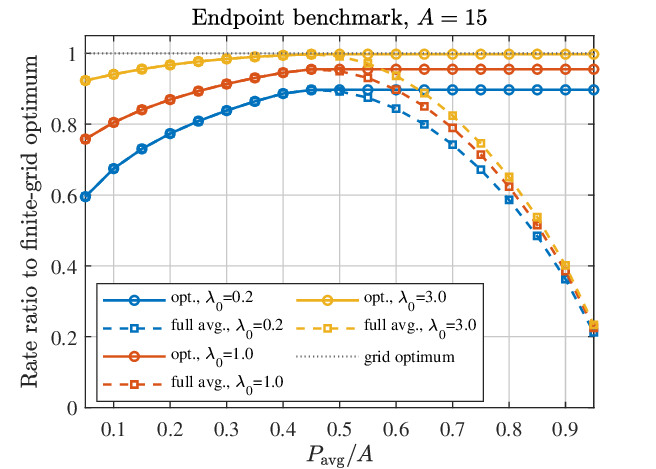}%
 \label{fig:endpoint_behavior_refined_b}}
 \caption{Endpoint-mass and endpoint-benchmark behavior in a two-tap ISI example with $A=15$, $h=(0.65,0.35)$, and $\lambda_0\in\{0.2,1.0,3.0\}$. The optimized endpoint benchmark maximizes over endpoint laws satisfying the average constraint, while the full-average endpoint benchmark enforces $p(A)=P_{\rm avg}/A$.}
 \label{fig:endpoint_behavior_refined}
\end{figure}

% ========================= NEW REVISION: NUMERICAL RELIABILITY CHECKS =========================
% Numerical sensitivity figures are inserted in the existing section;
% Fig. 5 is a single relative-rate-error figure (no subsection added).

{
Fig.~\ref{fig:block_length_robustness} evaluates block-length sensitivity using the percentage relative rate error $\varepsilon_{N_b}(\lambda_0)\triangleq 100\bigl|\widehat{\mathcal I}_{N_b}(\lambda_0)-\widehat{\mathcal I}_{5}(\lambda_0)\bigr|/\bigl|\widehat{\mathcal I}_{5}(\lambda_0)\bigr|\,\%$, with $N_b=5$ as the finite-block reference. For all three dark-current levels, the relative error drops sharply from $N_b=1$ to $N_b=2$ and is close to zero for $N_b\geq3$; the zero at $N_b=5$ follows by definition. These results support the numerical choice $N_b=3$ in the tested settings, but comparison with a finite reference block length does not establish an error bound relative to the infinite-history information rate.
}
\begin{figure}[!t]
 \centering
 \includegraphics[width=0.6\textwidth]{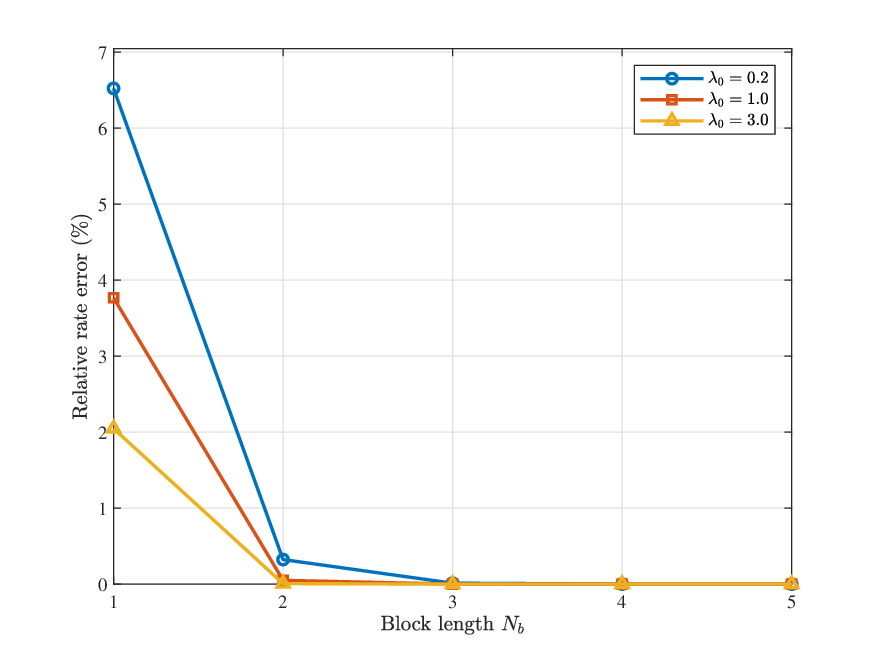}
 \caption{Percentage relative finite-grid rate error versus forward-recursion block length $N_b\in\{1,2,3,4,5\}$, with $N_b=5$ as the reference, for $A=2$, $P_{\rm avg}=0.5A$, and $\lambda_0\in\{0.2,1.0,3.0\}$. This diagnostic uses five SQP initializations.}
 \label{fig:block_length_robustness}
\end{figure}

{
Fig.~\ref{fig:grid_resolution_robustness} evaluates input-grid sensitivity through the relative rate error $\varepsilon_M(\rho)\triangleq 100\bigl|\widehat{\mathcal I}_{M}(\rho)-\widehat{\mathcal I}_{41}(\rho)\bigr|/\bigl|\widehat{\mathcal I}_{41}(\rho)\bigr|\,\%$, where $\rho=P_{\rm avg}/A$ and $M_{\rm ref}=41$ is the finest tested reference grid. The errors decrease rapidly as the grid is refined for all three average-to-peak ratios and are already close to zero near $M=31$; the zero error at $M=41$ follows from the reference definition. Therefore, the results support the numerical adequacy of $M=31$ in these tested settings, but the displayed relative errors are comparisons against a finite reference grid and do not certify convergence to the continuous-alphabet optimum or global SQP optimality.
}
\begin{figure}[!t]
 \centering
 \includegraphics[width=0.6\textwidth]{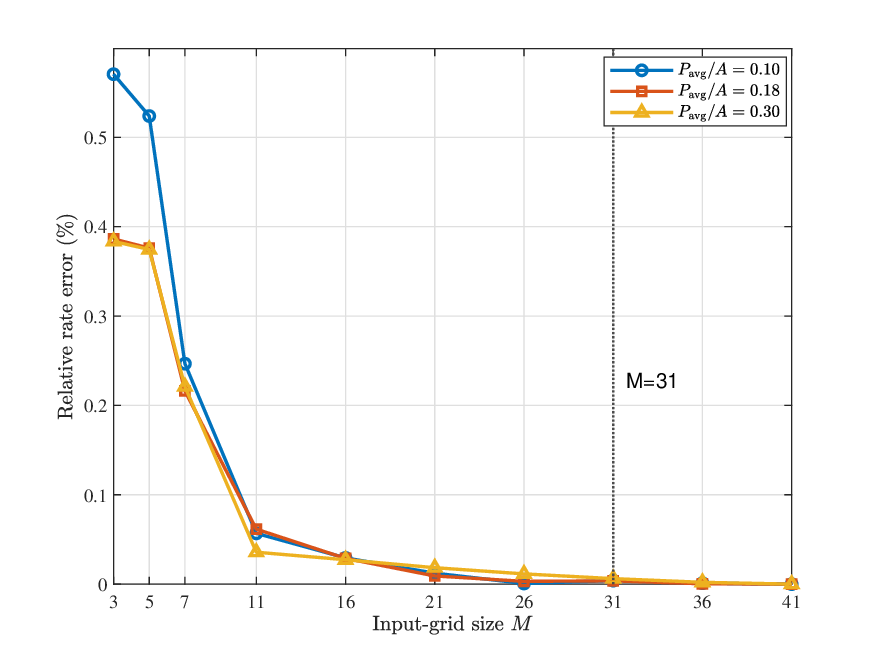}
 \caption{Relative finite-grid rate error with respect to $M_{\rm ref}=41$ for $M\in\{3,5,7,11,16,21,26,31,36,41\}$, $A=15$, $\lambda_0=1$, and $P_{\rm avg}/A\in\{0.10,0.18,0.30\}$. The dotted vertical line marks $M=31$; other numerical settings follow Table~\ref{tab:numerical_parameters}.}
 \label{fig:grid_resolution_robustness}
\end{figure}

{
Fig.~\ref{fig:truncation_robustness} examines photon-count truncation through the percentage relative finite-grid rate error $\varepsilon_{N_{\max}}(A)\triangleq 100\bigl|\widehat{\mathcal I}_{N_{\max}}(A)-\widehat{\mathcal I}_{50}(A)\bigr|/\bigl|\widehat{\mathcal I}_{50}(A)\bigr|\,\%$, using $N_{\max}=50$ as the reference and $N_{\max}\in\{10,20,30,40,50\}$. The relative errors for $A=8$ and $A=11$ are appreciable at $N_{\max}=10$ but rapidly decrease as the cutoff increases; the zero at $N_{\max}=50$ follows from the reference definition. The colored dotted lines mark the cutoffs $35$, $39$, and $47$ prescribed by \eqref{eq:numerical_truncation_rule}; together with the absolute rate-difference bounds in Table~\ref{tab:truncation_error_bounds}, the observed stabilization supports these cutoffs in the tested finite-grid experiments.
}
\begin{figure}[!t]
 \centering
 \includegraphics[width=0.6\textwidth]{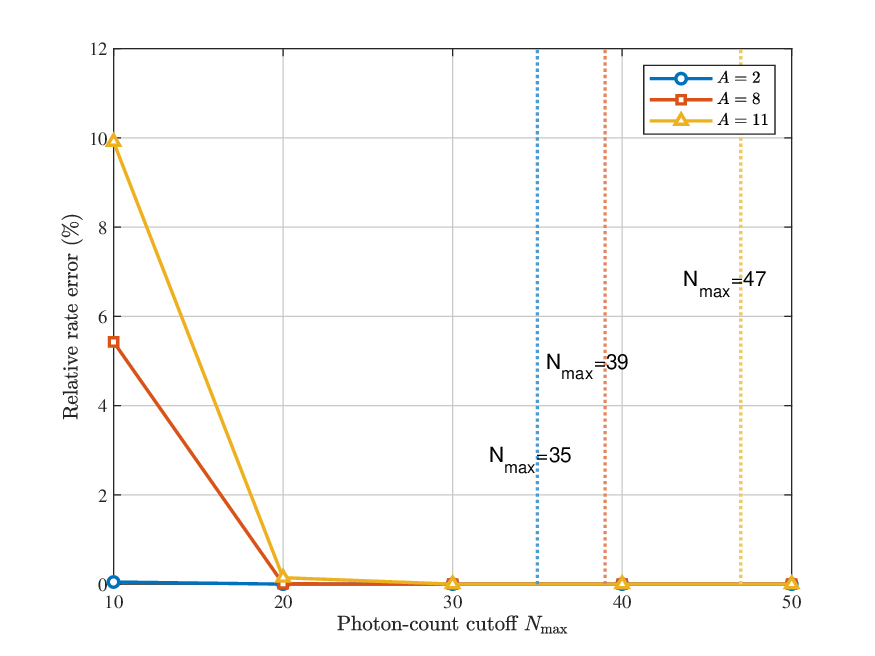}
 \caption{Percentage relative finite-grid rate error versus photon-count cutoff $N_{\max}\in\{10,20,30,40,50\}$, using $N_{\max}=50$ as the reference, for $A\in\{2,8,11\}$, $P_{\rm avg}=0.5A$, and $\lambda_0=1$. The color-matched dotted lines indicate the prescribed cutoffs $N_A=35,39,47$, respectively; other settings follow Table~\ref{tab:numerical_parameters}.}
 \label{fig:truncation_robustness}
\end{figure}

% ================= NEW REVISION: QUANTITATIVE TRUNCATION COMPARISON =================
{
To quantify the photon-count truncation effect in Fig.~\ref{fig:truncation_robustness},
we compare the prescribed cutoff $N_A$ from \eqref{eq:numerical_truncation_rule}
with $N_{\max}=50$ using the same finite-grid objective.
Write $q_N(\mu)\triangleq\Pr\{\operatorname{Poisson}(\mu)>N\}$ and
$\mu_{\max}=\lambda_0+A$. For the individually renormalized Poisson
kernels on their common alphabet $\{0,\ldots,50\}$, define
\[
\delta_A\triangleq
\frac{q_{N_A}(\mu_{\max})-q_{50}(\mu_{\max})}
     {1-q_{50}(\mu_{\max})}.
\]
For every $\mu\in[\lambda_0,\mu_{\max}]$, this quantity bounds
$\frac12\lVert Q_{N_A}(\cdot\mid\mu)-Q_{50}(\cdot\mid\mu)\rVert_{\rm TV}$,
where $Q_N$ denotes the renormalized Poisson kernel. The factor $1/2$
is required by the total-variation norm convention of Appendix~A.
Indeed, the conditional Poisson family truncated at $50$ is stochastically
increasing in its mean, so the maximal discrepancy occurs at $\mu_{\max}$.
Coupling both truncated channels using the same input process bounds one half of the total-variation norm of their $k$-symbol output laws by $k\delta_A$.
Applying the finite-alphabet entropy-continuity inequality to the
length-three and length-two output entropies, as well as the
single-symbol conditional entropy, gives, uniformly over every feasible
finite-grid input law $F$,
\[
\begin{aligned}
&\bigl|\widehat{\mathcal I}_{3,N_A}(F)
       -\widehat{\mathcal I}_{3,50}(F)\bigr|\le \mathcal B_{\rm tr}(A),\\
&\mathcal B_{\rm tr}(A)\triangleq\sum_{k=1}^{3}
 \bigl[k\delta_A\log(51^k-1)+h_2(k\delta_A)\bigr],
\end{aligned}
\]
where $h_2(t)=-t\log t-(1-t)\log(1-t)$ and
$\widehat{\mathcal I}_{3,N}(F)$ denotes the objective in
\eqref{eq:numerical_block_rate} evaluated with photon-count cutoff $N$.
The same bound therefore holds for the \emph{difference between the
respective globally optimized finite-grid values} at $N_A$ and $50$.
Table~\ref{tab:truncation_error_bounds} reports the maximal single-symbol
Poisson tails and the corresponding rate-gap upper bounds.
All three bounds are below $10^{-11}$ nats/use, supporting the prescribed
cutoffs for these finite-block comparisons. These are uniform objective-
truncation bounds relative to $N_{\max}=50$, not empirical SQP gaps or
certificates for the infinite-history approximation.
}
\begin{table}[!t]
 \centering
 \caption{Poisson tail probabilities and upper bounds on the finite-grid
 rate difference between the prescribed cutoff and $N_{\max}=50$;
 the other parameters follow Fig.~\ref{fig:truncation_robustness}.}
 \label{tab:truncation_error_bounds}
 \footnotesize
 \begin{tabular}{@{}cccc@{}}
 \toprule
 $A$ & $N_A$ (rule) & $q_{N_A}(\mu_{\max})$ & $\mathcal B_{\rm tr}(A)$ (nats/use) \\
 \midrule
 $2$ & $35$ & $2.19\times10^{-26}$ & $8.98\times10^{-24}$ \\
 $8$ & $39$ & $2.86\times10^{-14}$ & $6.96\times10^{-12}$ \\
 $11$ & $47$ & $4.13\times10^{-15}$ & $1.04\times10^{-12}$ \\
 \bottomrule
 \end{tabular}
\end{table}
% ================ END QUANTITATIVE TRUNCATION COMPARISON =================

{
Fig.~\ref{fig:multistart_robustness} shows the empirical cumulative fraction of SQP runs whose final rate differs from the best value found across 40 initializations by at most a specified absolute gap. In the displayed experiment, $60\%$ and $95\%$ of the starts are within $10^{-6}$ and $10^{-4}$ nats/use of the best-found value, respectively, while a few starts lead to noticeably lower values. These results motivate the multi-start strategy used for finite-grid optimization.
}
\begin{figure}[!t]
 \centering
 \includegraphics[width=0.6\textwidth]{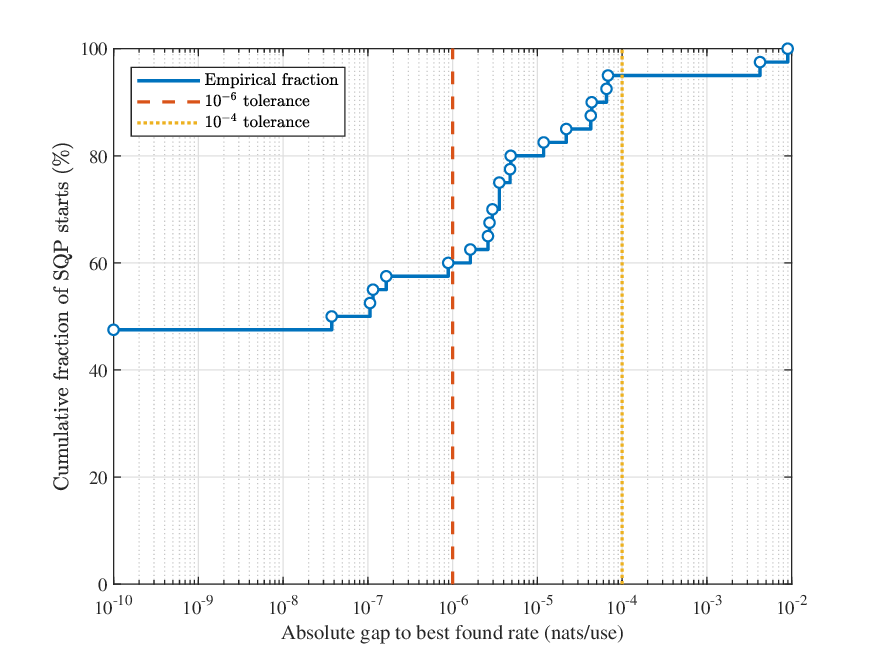}
 \caption{Empirical cumulative distribution of absolute SQP objective gaps over 40 initializations for $A=8$, $P_{\rm avg}=0.5A$, and $\lambda_0=1$. Gaps below $10^{-10}$ are displayed at $10^{-10}$ on the logarithmic axis.}
 \label{fig:multistart_robustness}
\end{figure}
% ======================= END NEW NUMERICAL REVISION =======================

The numerical results lead to three observations for the discretized i.i.d. optimization. First, the optimized finite-grid distributions are sparse and often strongly concentrated on the endpoints. Second, the active grid set and the finite-grid rate estimate are sensitive to the dark current, peak constraint, and ISI tap profile. Third, optimized endpoint signaling can be close to the finite-grid optimum in the tested two-tap examples, while forcing the endpoint law to use the full average power can be suboptimal. These observations are consistent with the finite-support structural result in Theorem~\ref{thm:finite_support}, but they rely on finite-grid and finite-block approximations and are not an independent verification of the continuous-alphabet theorem.

\section{Conclusion}\label{sec:conclusion}

In this paper, we studied the stationary mutual-information-rate optimization under an i.i.d. input restriction for finite-memory Poisson channels under joint peak and average optical-intensity constraints. By working directly with the stationary mutual information rate, we proved that every optimizing marginal input law has finite support, although the ISI channel induces a continuous-state hidden Markov output process. The key step is a uniform filtering-forgetting estimate: after a point-mass perturbation affects a finite physical output window, its remaining influence is propagated only through predictive filters whose dependence on remote initial conditions decays exponentially. This estimate justifies the entropy-rate first variation and leads to a holomorphic KKT defect function. The supralinear growth of this defect on the positive real axis then rules out infinitely many support points. Therefore, within the i.i.d. input class, the search for rate-maximizing marginal laws is structurally restricted to finite-support signaling distributions.

\appendices

\section{Measure-Theoretic Conventions}
\label{app:measure_tools}

We collect the measure-theoretic conventions used in the proof. Let $(\mathcal S,\mathcal A)$ be a measurable space. If a finite measure $\mu$ is absolutely continuous with respect to another finite measure $\nu$, written $\mu\ll\nu$, the Radon-Nikodym theorem states that there exists a measurable function $d\mu/d\nu$ such that
\begin{equation}
\label{eq:rn_definition_app}
\mu(E)=\int_E\frac{d\mu}{d\nu}(s)d\nu(s),
\qquad E\in\mathcal A .
\end{equation}
The function $d\mu/d\nu$ is the Radon-Nikodym density of $\mu$ with respect to $\nu$~\cite{folland1999real}. In Lemma~\ref{lem:uniform_filter_forgetting}, the relevant measures are block kernels with common reference measure $F^{\otimes N}$; no density with respect to Lebesgue measure is required.

For a finite signed measure $\eta$, we use the total-variation norm
\begin{equation}
\label{eq:tv_signed_partition}
\lVert\eta\rVert_{\rm TV}
\triangleq
\sup_{\{E_i\}}
\sum_i|\eta(E_i)|,
\end{equation}
where the supremum is over all finite measurable partitions. If $\eta=\eta^+-\eta^-$ is the Jordan decomposition, then
\begin{equation}
\lVert\eta\rVert_{\rm TV}=\eta^+(\mathcal S)+\eta^-(\mathcal S).
\end{equation}
For a finite complex measure $\nu$, we use $\lVert\nu\rVert_{\rm TV}\triangleq|\nu|(\mathcal S)$, where $|\nu|$ is its total-variation measure. In particular, $\lVert\operatorname{Re}\nu\rVert_{\rm TV}$ and $\lVert\operatorname{Im}\nu\rVert_{\rm TV}$ are bounded by $\lVert\nu\rVert_{\rm TV}$.
For probability measures $\mu$ and $\nu$, the signed measure $\eta=\mu-\nu$ has zero total mass. Hence $\eta^+(\mathcal S)=\eta^-(\mathcal S)$ and
\begin{equation}
\label{eq:tv_probability_factor_two}
\lVert\mu-\nu\rVert_{\rm TV}
=2\sup_{E\in\mathcal A}|\mu(E)-\nu(E)|.
\end{equation}
This is the convention used throughout the paper; some probability texts use the right-hand side of \eqref{eq:tv_probability_factor_two} divided by two.

If $(U,V)$ is any coupling of $\mu$ and $\nu$, then
\begin{equation}
\label{eq:tv_coupling_bound}
\lVert\mu-\nu\rVert_{\rm TV}
\le 2\mathsf P(U\ne V).
\end{equation}
A maximal coupling attains equality in \eqref{eq:tv_coupling_bound}; see, e.g.,~\cite{lindvall2002lectures}.

{
\section{Hilbert-Metric Facts for Positive Kernels}
\label{app:hilbert_metric}
We collect the projective-metric facts used in Lemma~\ref{lem:uniform_filter_forgetting}; see, e.g.,~\cite{legland2004hilbert}. Let $\mathcal M_+(S)$ denote the cone of finite nonnegative measures on a measurable space $S$. For two nonzero comparable measures $\xi,\zeta\in\mathcal M_+(S)$, define the Hilbert projective metric
\begin{equation}
 h(\xi,\zeta)\triangleq\log\frac{b_\star}{a_\star},
\end{equation}
where $a_\star$ and $b_\star$ are the optimal constants satisfying $a_\star\zeta\le\xi\le b_\star\zeta$; set $h(\xi,\zeta)=+\infty$ when the measures are not comparable. The metric is invariant under positive rescaling. If $\bar\xi=\xi/\xi(S)$, then
\begin{equation}
\label{eq:hilbert_to_tv_appendix}
\lVert\bar\xi-\bar\zeta\rVert_{\rm TV}
\le\frac{2}{\log3}h(\xi,\zeta).
\end{equation}

A positive kernel $K$ is called $\varepsilon$-mixing if there is a nonzero finite measure $\Lambda$ such that
\begin{equation}
\label{eq:epsilon_mixing_appendix}
\varepsilon\Lambda(E)\le K(x,E)\le\varepsilon^{-1}\Lambda(E)
\end{equation}
for every $x$ and measurable $E$. For such a kernel,
\begin{equation}
\label{eq:first_kernel_regularization_appendix}
h(\xi K,\zeta K)
\le
\varepsilon^{-2}\lVert\bar\xi-\bar\zeta\rVert_{\rm TV},
\end{equation}
so one mixing kernel regularizes arbitrary nonzero initial measures even when their original Hilbert distance is infinite. Its projective diameter satisfies $H(K)\le\log(\varepsilon^{-4})$, and the Birkhoff contraction theorem gives
\begin{equation}
\label{eq:birkhoff_coeff_appendix}
\tau(K)
\le
\tanh\!\left(\frac{H(K)}4\right)
\le
\frac{1-\varepsilon^2}{1+\varepsilon^2}<1.
\end{equation}
Consequently, for a sequence of mixing kernels $K_1,\ldots,K_j$ with constants $\varepsilon_1,\ldots,\varepsilon_j$ and contraction coefficients $\tau_1,\ldots,\tau_j$,
\begin{equation}
\label{eq:hilbert_product_tv_appendix}
\left\|
\overline{\mu K_1\cdots K_j}
-
\overline{\nu K_1\cdots K_j}
\right\|_{\rm TV}
\le
\frac{2}{\log3}\varepsilon_1^{-2}
\prod_{k=2}^{j}\tau_k
\lVert\mu-\nu\rVert_{\rm TV},
\end{equation}
where the overline denotes normalization to unit mass. This estimate is deterministic once the realized kernels are fixed; a probability law for the observations is needed only when it is subsequently averaged.

For incomplete blocks we also use the following elementary normalization bound. Let $K$ be a positive kernel with
\begin{equation}
0<m\le K(x,S')\le M<\infty
\end{equation}
for every $x$. Then, for probability measures $\mu,\nu$,
\begin{equation}
\label{eq:normalization_lipschitz_appendix}
\left\|
\frac{\mu K}{(\mu K)(S')}
-
\frac{\nu K}{(\nu K)(S')}
\right\|_{\rm TV}
\le
2\frac{M}{m}\lVert\mu-\nu\rVert_{\rm TV}.
\end{equation}
Indeed, $\lVert(\mu-\nu)K\rVert_{\rm TV}\le M\lVert\mu-\nu\rVert_{\rm TV}$ and $|(\mu K)(S')-(\nu K)(S')|\le M\lVert\mu-\nu\rVert_{\rm TV}$; division by the smaller normalization constant proves \eqref{eq:normalization_lipschitz_appendix}.
}
\section{Auxiliary Analytic and Convex-Analytic Facts}
\label{app:aux_theorems}

The proof uses three standard facts. First, the identity theorem states that if a holomorphic function on a connected open set has zeros with an accumulation point inside the set, then the function is identically zero on that set~\cite{ahlfors1979complex}. Second, the Weierstrass theorem states that a locally uniformly convergent series of holomorphic functions has a holomorphic sum~\cite{ahlfors1979complex}. Third, for a compact metric space $\mathcal A$, the set $\mathcal P(\mathcal A)$ of Borel probability measures is weakly compact~\cite{billingsley1999convergence}.

For the KKT step, the following convex-analytic form is used. Let $\mathcal C$ be a convex subset of a locally convex topological vector space, let $g$ be a continuous affine functional, and suppose that the constraint $g(x)\le0$ satisfies Slater's condition: there exists $x_0\in\mathcal C$ with $g(x_0)<0$. If a continuous linear functional $L$ is maximized over $\{x\in\mathcal C:g(x)\le0\}$ at $x^\times$, then there exists $\lambda\ge0$ such that
\begin{equation}
\label{eq:kkt_appendix_variational}
L(x)-\lambda g(x)
\le
L(x^\times)-\lambda g(x^\times),
\qquad x\in\mathcal C,
\end{equation}
and $\lambda g(x^\times)=0$~\cite{rockafellar1970convex}. In Lemma~\ref{lem:kkt}, $\mathcal C=\mathcal P([0,A])$ and $g(F)=\int x\,dF(x)-P_{\rm avg}$.

\bibliographystyle{IEEEtran}
\bibliography{ref_tit}

\end{document}